\documentclass[sigconf]{aamas} 

\usepackage{balance} 

\setcopyright{ifaamas}
\acmConference[AAMAS '24]{Proc.\@ of the 23rd International Conference
on Autonomous Agents and Multiagent Systems (AAMAS 2024)}{May 6 -- 10, 2024}
{Auckland, New Zealand}{N.~Alechina, V.~Dignum, M.~Dastani, J.S.~Sichman (eds.)}
\copyrightyear{2024}
\acmYear{2024}
\acmDOI{}
\acmPrice{}
\acmISBN{}

\usepackage{graphicx}
\usepackage{amsfonts}  
\usepackage{amsmath}  
\usepackage[linesnumbered,ruled,vlined,noend]{algorithm2e}
\usepackage{mathtools} 
\usepackage{xcolor} 

\usepackage{csquotes}
\usepackage{subcaption}

\newcommand{\argmin}{\operatornamewithlimits{argmin}}

\makeatletter
\newcommand{\pushright}[1]{\ifmeasuring@#1\else\omit$\displaystyle#1$\ignorespaces\fi}
\makeatother

\newcommand{\BibTeX}{\rm B\kern-.05em{\sc i\kern-.025em b}\kern-.08em\TeX}

\newtheorem*{observation}{Observation}

\newtheorem{imr}{}
\newenvironment{cmr}[1]
{\renewcommand\theimr{#1}\imr}
{\endimr}

\newcommand{\R}{\ensuremath{\mathbb{R}}}

\newcommand{\K}{\ensuremath{\mathcal{K}}}
\newcommand{\floor}[1]{\lfloor #1 \rfloor}

\newcommand{\valueDistribution}{\mathcal{F}}
\newcommand{\utopianGap}{\ensuremath{\mathcal{G}}}

\renewcommand{\S}{\ensuremath{\mathbb{S}^n}}
\newcommand{\C}{\ensuremath{\mathbb{C}^n}}

\def\absolute#1{\left\lvert #1 \right\rvert }

\newcommand{\hoursFiveOptimal}[0]{60 }
\newcommand{\hoursFiveOptimalAll}[0]{200 }
\newcommand{\timeSixOptimal}[0]{over~100~years}
\newcommand{\numSamplesStd}[0]{3,000 }
\newcommand{\numSamplesFive}[0]{10 }
\newcommand{\ppoTimesteps}[0]{2,000,000 }
\newcommand{\ppoIterations}[0]{40 }
\newcommand{\ppoTrajectories}[0]{4,096 }
\newcommand{\ppoEpochs}[0]{10 }

\acmSubmissionID{247}

\title[Reducing Optimism Bias in Incomplete Cooperative Games]{Reducing Optimism Bias in Incomplete Cooperative Games}

\author{Filip \'{U}radn\'{i}k}
\affiliation{
  \institution{Charles University}
  \city{Prague}
  \country{Czechia}}
\email{uradnik@kam.mff.cuni.cz}

\author{David Sychrovsk\'{y}}
\affiliation{
  \institution{Charles University}
  \city{Prague}
  \country{Czechia}}
\email{sychrovsky@kam.mff.cuni.cz}

\author{Jakub \v{C}ern\'{y}}
\affiliation{
  \institution{Columbia University}
  \city{New York}
  \country{USA}}
\email{cerny@disroot.org}

\author{Martin \v{C}ern\'{y}}
\affiliation{
  \institution{Charles University}
  \city{Prague}
  \country{Czechia}}
\email{cerny@kam.mff.cuni.cz}

\begin{abstract}
Cooperative game theory has diverse applications in contemporary artificial intelligence, including domains like interpretable machine learning, resource allocation, and collaborative decision-making. However, specifying a cooperative game entails assigning values to exponentially many coalitions, and obtaining even a single value can be resource-intensive in practice. Yet simply leaving certain coalition values undisclosed introduces ambiguity regarding individual contributions to the collective grand coalition. This ambiguity often leads to players holding overly optimistic expectations, stemming from either inherent biases or strategic considerations, frequently resulting in \textit{collective claims exceeding the actual grand coalition value}. In this paper, we present a framework aimed at optimizing the sequence for revealing coalition values, with the overarching goal of efficiently closing the gap between players' expectations and achievable outcomes in cooperative games. Our contributions are threefold: (i) we study the individual players' optimistic completions of games with missing coalition values along with the arising gap, and investigate its analytical characteristics that facilitate more efficient optimization; (ii) we develop methods to minimize this gap over classes of games with a known prior by disclosing values of additional coalitions in both offline and online fashion; and (iii) we empirically demonstrate the algorithms' performance in practical scenarios, together with an investigation into the typical order of revealing coalition values. 
\end{abstract}


\makeatletter
\gdef\@copyrightpermission{
	\begin{minipage}{0.3\columnwidth}
		\href{https://creativecommons.org/licenses/by/4.0/}{\includegraphics[width=0.90\textwidth]{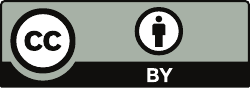}}
	\end{minipage}\hfill
	\begin{minipage}{0.7\columnwidth}
		\href{https://creativecommons.org/licenses/by/4.0/}{This work is licensed under a Creative Commons Attribution International 4.0 License.}
	\end{minipage}
	\vspace{5pt}
}
\makeatother

\begin{document}


\pagestyle{fancy}
\fancyhead{}


\maketitle


\section{Introduction}

Cooperative game theory provides a versatile framework for modeling coalition formation and collective payoff distributions in multiagent interactions. Its far-reaching applications span diverse fields, including supply chain management~\cite{nagarajan2008game}, communication networks~\cite{saad2009coalitional}, logistics and resource allocation~\cite{lozano2013cooperative}, or environmental agreements~\cite{finus2008game}.
However, beneath its promising facade lies a fundamental challenge -- specifying a cooperative game requires assigning a value to each possible coalition, i.e., a subset of players, which can be a daunting process as the number of coalitions is exponential in the number of interacting agents.

In practice, acquiring even a single value for a single coalition can be a resource-intensive endeavor. Take, for instance, the realm of machine learning, where determining the value of a feature subset that represents a coalition in the celebrated explainable approach SHAP~\cite{NIPS2017_7062} corresponds to retraining an entire model, consuming time and computational resources. 
What is more, these contributions can have ripple effects on subsequent financial outlays, such as acquiring new training samples.

In the corporate world, estimating an employee's contribution to collective performance may facilitate their fair evaluation~\cite{murphy1995understanding}, but obtaining the value of a coalition may involve the intricate process of rearranging teams of employees, incurring operational costs and potentially causing disruptions.

Yet, simply leaving the values of many coalitions undetermined further compounds the problem by opening the doors to ambiguity. When dealing with incomplete information, humans tend to exhibit a natural bias towards optimism, defined as the disposition to be overly optimistic about the probability of positive future events and to downplay the likelihood of negative future events~\cite{irwin1953stated,weinstein1980unrealistic,slovic2016facts,slovic87}.
For instance, optimism bias can lead individuals to underestimate the effort needed for retirement savings, potentially resulting in insufficient funds~\cite{puri2007optimism}, or to downplay the likelihood of catastrophic events like natural disasters, leading to inadequate preparedness and potentially life-threatening situations~\cite{weinstein1989optimistic}. In the context of cooperative games, where the players lack precise information about coalition values, such optimism can lead to inflated expectations. The individuals believe their contributions are more significant than they objectively are. This optimism translates into unrealistic demands for a larger share of the grand, i.e. all-agents, coalition value, even to the point where the sum of individual claims surpasses the actual value of the grand coalition.
For example, the companies might sometimes demand exorbitant prices for their data, just as employees may occasionally request wages that are unrealistically high.  Moreover, even when players do not possess inherent optimism biases, strategically adopting an optimistic facade can prove advantageous, influencing negotiations and outcomes in cooperative games. This discrepancy creates a critical gap, which we refer to as the \textit{cumulative utopian gap}, between what the players expect and what is feasible within the game.

To narrow this gap, we further assume the existence of an external principal. In the examples we mentioned, this role could be reserved for the company manager or the machine learning engineer. This principal possesses the unique ability to determine the sequence in which coalitions are revealed, and they are provided with a limited number of opportunities to exercise this control. In this manner, the principal mitigates the ambiguity within the system, thereby diminishing the utopian gap as a consequence. Importantly, we assume that each revelation step carries a roughly equivalent cost, ensuring that there is no inherent preference in which coalition is unveiled next besides its effect on bridging the gap. The primary objective of the principal hence reduces solely to minimizing this gap under the budget constraint. Although we do not explicitly outline the methodology for resolving situations when a non-zero gap exists, our underlying assumption is that a lower gap is favorable. In practical terms, a diminished gap may have implications for reducing the bargaining power of individual players, or in more extreme scenarios, it may limit the additional financial incentives required by the principal to encourage players to participate in the grand coalition.

\subsection{Organization and Contributions}

We begin by formally defining the framework of cooperative games, with a particular emphasis on the Shapley value\footnote{While our work draws from the properties of the Shapley value, it is important to note that our findings extend beyond the scope of cooperative game theory alone. For example, they have potential applicability within the theory of capacities~\cite{choquet1954theory}, where the Shapley value carries alternative interpretations. For a broader generalization of our approach, please refer to the related work section.}~\cite{Shapley1953} as a mechanism for fairly distributing the grand coalition value among the individual players. We then explain how to extend the framework to encompass incomplete cooperative games, which feature missing coalition values, a prevalent occurrence in real-world scenarios.

Afterward, we delve into our main contributions. Central to our study is the introduction of superadditive utopian games, designed to incorporate players' optimism biases, resulting in the emergence of the cumulative utopian gap. We establish fundamental theoretical properties of the utopian gap, including its monotonicity, additivity, and circumstances under which it becomes zero. Additionally, we offer an alternative geometric interpretation of the utopian gap, viewing it as a quantification of the remaining uncertainty within the space of potential true cooperative games, considering the coalition values that have been determined thus far.

Building upon these foundations, we formulate both offline and online problems for the principal aiming to minimize the utopian gap, offering a suite of heuristic and approximative algorithms for each scenario. Our empirical analysis provides valuable insights into algorithms' performance and the coalitions typically revealed early in the online and offline setups. Importantly, our findings also illustrate the non-linear nature of the utopian gap's decrease with the number of revealed coalition values, offering nuanced perspectives on what a principal can expect at various stages of the revelation process. Our results further indicate that for specific classes of monotone supermodular games, the gap can be nearly completely reduced by revealing just $\mathcal{O}(n)$ coalitions. This stands in contrast to the general requirement of exploring all possible values to minimize it to zero, as demonstrated in Appendix~\ref{app:zero-gap-not-attainable}.

\subsection{Related Work}

To the best of our knowledge, there is no existing research directly addressing the reduction of ambiguity in the context of cooperative game theory. Nevertheless, the process of querying coalition values can be seen as an online construction of a compact function representation, an approach supported by promising results in the field of cooperative game theory~\cite{Masuya2016,Bok2023,Cerny2023,CernyGrabisch2023}. These results indicate the feasibility of such an approach, with some efforts demonstrating a substantial exponential reduction in the number of values required to represent superadditive functions. A comprehensive survey and detailed presentation of many of these findings can be found in Chalkiadakis' book~\cite{Chalkiadakis2012}. While tailored representations have achieved significant reductions in specific cases~\cite{ieong2005marginal}, no general approach for constructing such representation when given a subclass of games has been identified.

Our work adopts an approach reminiscent of active learning~\cite{settlet2009active}, where an algorithm actively queries an oracle for labels to new data points to construct the most informative dataset, particularly in situations where labeling is resource-intensive. In our context, we seek coalition values that minimize the utopian gap. To approximate the optimal querying strategy based on this concept, we employ reinforcement learning~\cite{sutton2014reinforcement}.

\section{Preliminaries}
Here, we present fundamentals of cooperative games (the reader is encouraged to see~\cite{Peleg2007} for a more extensive introduction). First, we define cooperative games, introduce their classes and the Shapley value. Then we talk about the generalized model of incomplete games where only some values of the game are known. We define sets of \emph{$\S$-extensions} and the \emph{lower}/\emph{upper} \emph{games}. These notions delimit possible values of $\S$-extensions of incomplete games.

\begin{definition}
	\label{def: cooperative game}
	A \emph{cooperative game} is an ordered pair $(N,v)$ where $N = \{1,\ldots n\}$ and $v\colon 2^N \to \mathbb{R}$ is the characteristic function of the cooperative game. Further, $v(\emptyset) = 0$.
\end{definition}

It is convenient to view the characteristic function $v$ as a point in $2^n$-dimensional space.
We say a cooperative game $(N,v)$ is
\begin{itemize}
	\item {\it additive} if for every $S \subseteq N$,
	      \vspace{-0.1cm}
	      \begin{equation}
		      v(S) = \sum_{i \in S}v(\{i\}),\vspace{-0.1cm}
	      \end{equation}

	\item {\it superadditive} if for every $S, T \subseteq N, S\cap T=\emptyset$,
	      \vspace{-0.1cm}
	      \begin{equation}\label{eq: superaddive constraints}
		      v(S) + v(T) \le v(T\cup S)\vspace{-0.1cm}
	      \end{equation}

	\item {\it supermodular} if  for every $S\subseteq N\setminus \{i,j\}$,
	      \vspace{-0.1cm}
	      \begin{equation}\label{eq: supermodular constraints}
		      v(S\cup \{j\}) - v(S) \le v(S \cup \{i, j\}) - v(S \cup \{i\}).
	      \end{equation}

\end{itemize}
By $\S$ and $\C$, we denote the sets of superadditive, respectively supermodular games on $\lvert N \rvert = n$ players. Further we refer to $S\subseteq N$ as a \emph{coalition}
and $N$ as the \emph{grand coalition}.
The following definition contains one of the most studied solution concepts in cooperative game theory, which models fair distribution of the value of the grand coalition~\cite{Shapley1953}.

\noindent

\begin{definition}
	The {\it Shapley value}  $\phi: \mathbb{R}^{2^n} \to \mathbb{R}^n$ of a cooperative game $(N,v)$ is defined for $i \in N$ as
	\begin{equation}\label{eq:shapley}
		\phi_i(v) \coloneqq
		\sum_{S\subseteq N\setminus \{i\}}
		\frac{|S|!(|N|-|S|-1)!}{|N|!}\bigl(v(S\cup\{i\}) - v(S)\bigr).
	\end{equation}
\end{definition}

The Shapley value is \textit{linear} and \textit{efficient}, which captures the fact that exactly $v(N)$ is distributed among the players. Specifically,
\begin{equation}\label{eq:efficiency}
	\sum_{i\in N}\phi_i (v) = v(N).
\end{equation}

The characteristic function $v$ is represented by $2^n$ real values.
In many applications, however, obtaining all the values might be too expensive.
To deal with this problem, $\K \subseteq 2^N$ is introduced in the definition of an \textit{incomplete game} which represents only coalitions with known value.
In our application, we model obtaining new information about the unknown (but well defined) values by extending the set $\K$.
Therefore, we view $\K$ as a `masking set' applied to some complete game and use the characteristic function of that complete game in our definition.\footnote{This is in contrast with the standard definition where only $v:\K\to\mathbb{R}$ is used.}
Another distinction from the standard definition is that we require knowledge of at least \textit{minimal information}\footnote{
	In the standard definition, only $\emptyset\in\K$ is required~\cite{Masuya2016,Bok2023}.}
\begin{equation}
	\K_0 \coloneqq \{\emptyset,N\} \cup \{\{i\} \mid i \in N\}.
\end{equation}

\begin{definition}
	An \textit{incomplete cooperative game} is $(N,\K,v)$ where $N = \{1,\dots,n\}$, $\K_0 \subseteq \K \subseteq 2^N$, $v\colon2^N \to \mathbb{R}$ is the characteristic function, and $v(\emptyset)=0$.
	We say that $(N, v)$ is the \textit{underlying game} of $(N,\K,v)$.
	Further, an incomplete cooperative game $(N,\K,v)$ is \textit{minimal} if $\K=\K_0$.
\end{definition}

If further properties of the underlying game are assumed, one may impose restrictions on values of $S \notin \K$ even though they are not known exactly. The property we consider is superadditivity, which is natural in many situations and common throughout the literature on exchange economies, cost-distribution problems, or supermodular optimization~\cite{Masuya2016, Peleg2007}.\footnote{{The immediacy of this property is not always guaranteed. For instance, SHAP may not consistently yield a superadditive game. However, within this particular context, the property is observed, for example, in uncorrelated models within an ensemble.}} Under this assumption and based on partial knowledge represented by $\K$, one can describe the set of \textit{candidates} for the underlying game. The set consists of extensions of the partial set function, which satisfy superadditivity.
\begin{definition}
	Let $(N, \K, v)$ be an incomplete cooperative game. Then $(N,w)$ is a \emph{$\S$-extension} of $(N,\K,v)$ if $(N,w) \in \S$ and
	\begin{equation}
		v(S) = w(S), \hspace{5ex} S\in \K.
	\end{equation}
	We say $(N,\K,v)$ is \emph{$\S$-extendable} if it has a $\S$-extension and we denote the set	of $\S$-extensions by $\S(\K,v)$.
\end{definition}

Since the set of $\S$-extensions is given by a system of linear inequalities, it forms a convex polyhedron in $\mathbb{R}^{2^n}$.
Furthermore, assuming non-negativity of the values, set $\S(\K,v)$ can be tightly enclosed by an hyper-rectangle given by the so called \emph{lower}/\emph{upper} games~\cite{Masuya2016}. Specifically, the \textit{lower game} $(N,\underline{v})$ of $(N,\K,v)$ is
\begin{equation}\label{eq: lower game}
	\underline{v}_{\K}(S) \coloneqq \max\limits_{\substack{S_1,\dots,S_k \in \K \\
			\bigcup_i S_i = S \\
			S_i \cap S_j = \emptyset}}
	\sum_{i=1}^kv(S_i),
\end{equation}
and the \textit{upper game} $(N,\overline{v})$ of $(N,\K,v)$ is
\begin{equation}
	\label{eq: upper game}
	\overline{v}_{\K}(S) \coloneqq \min\limits_{T \in \K: S \subseteq T}v(T) - \underline{v}(T \setminus S).
\end{equation}

Within this specific context, we restate the following established results that will hold importance for our later findings.

\begin{theorem}\label{lem:super bounds}
	Let $(N,\K,v)$ be an $\S$-extendable incomplete game with non-negative values. Then for every $\S$-extension $(N,w)$ of $(N,\K,v)$ it holds
	\begin{equation}
		\underline{v}_{\K}(S) \leq w(S) \leq \overline{v}_{\K}(S), \hspace{5ex} \forall S\subseteq N.
	\end{equation}
	Further, $\forall S \notin \K$, there are $\S$-extensions $(N,w_1), (N,w_2)$ such that
	\[
		w_1(S) = \underline{v}_{\K}(S) {\rm \hspace{2ex} and\hspace{2ex}} w_2(S) = \overline{v}_{\K}(S).
	\]
\end{theorem}
\begin{proof}
	The first part of the theorem is equivalent to Theorem 1 in~\cite{Masuya2016}. The second part follows from Theorem 3 in~\cite{Masuya2016}, which states that any cooperative game with a characteristic function defined for a non-empty coalition $S$ as
	\begin{equation}\label{eq: utopian value function}
		v^S(T) \coloneqq \begin{cases}
			\overline{v}_{\K}(T)  & S\subseteq T,    \\
			\underline{v}_{\K}(T) & S\not\subseteq T \\
		\end{cases}
	\end{equation}
	is an $\S$-extension. For $S \notin \K$, choose $w_1=v^N$ and $w_2=v^S$.
\end{proof}

Note that the hyper-rectangle given by lower/upper games might contain non-superadditive games.

\section{Optimism Bias and Its Minimization}\label{sec:optimism-bias}

The uncertainty arising from $\K$ in an incomplete cooperative game allows each player to speculate about their payment. As players experience the bias towards optimism, they can consider every $\S$-extension of the incomplete game and choose the one which benefits them the most. Such extension can be easily derived based on the properties of the Shapley value.

\begin{definition}\label{def:player-preferred-games}
	Let $(N,\K,v)$ be a $\S$-extendable incomplete cooperative game and $i \in N$.
	Then player $i$'s \emph{utopian game} is $(N,v_i)$, where
	\begin{equation}
		v_i(S) \coloneqq \begin{cases}
			\overline{v}_{\K}(S)  & i \in S,    \\
			\underline{v}_{\K}(S) & i \notin S. \\
		\end{cases}
	\end{equation}
\end{definition}

\begin{proposition}
	Let $(N,\K,v)$ be a $\S$-extendable incomplete game.
	Then the utopian game $(N,v_i) \in \S(\K,v)$ for every $i \in N$.
\end{proposition}
\begin{proof}
	From~\eqref{eq: lower game}, \eqref{eq: upper game}, it is immediate $(N,v_i)$ extends $(N,\K,v)$. Further, as mentioned earlier, Theorem 3 in~\cite{Masuya2016} states that any cooperative game with a characteristic function \eqref{eq: utopian value function}
	is superadditive. We obtain the result by setting $S=\{i\}$.
\end{proof}

When each player demands the payoff given to them in their utopian game, it leads to an efficiency violation due to the players' cumulative demands surpassing the value of $v(N)$. We refer to the disparity between the players' expectations and the attainable outcome as the {\it cumulative utopian gap}.

\begin{definition}
	\label{def: exploitability}
	Let $(N,\hat{\K},v)$ be a $\S$-extendable incomplete cooperative game and $\K = \hat{\K} \setminus \K_0$. Then the \textit{cumulative utopian gap} of $(N,\hat{\K},v)$ is
	\begin{equation}
		\label{eq: utopian gap}
		\utopianGap_{(N,v)}(\K) \coloneqq \sum_{i\in N} \left(\max\limits_{w\in \S(\hat{\K},v)} \phi_i(w)\right) - v(N).
	\end{equation}
\end{definition}
\noindent

Due to the properties of utopian games, the cumulative utopian gap can be expressed in terms of individual $v_i$'s. As a consequence, one can express the gap as an affine combination of difference between the values of the upper and the lower functions.

\begin{proposition}\label{prop:gap-is-utopian}
	\label{lem:preferred-game}
	Let $(N,\hat{\K},v)$ be an $\S$-extendable incomplete game. Then the cumulative utopian gap is
	\begin{equation}
		\utopianGap_{(N,v)}(\K) = \sum_{i \in N}\phi_i(v_i) - v(N) = \sum_{S \subseteq N}\alpha_S\Delta_{\hat{\K}}(S),
	\end{equation}
	where $\alpha_S = \frac{|S|!(|N| - |S|)!}{|N|!}$ and $\Delta_{\hat{\K}}(S) = \overline{v}_{\hat{\K}}(S) - \underline{v}_{\hat{\K}}(S)$.
\end{proposition}
\begin{proof}[Proof Sketch]
	Notice from~\eqref{eq:shapley} and the definition of $\S(\hat{\K},v)$ that $\phi_i(w)$ is bounded from above by
	\begin{equation*}
		\sum_{S \subseteq N \setminus \{i\}}\frac{|S|!(|N|-|S|-1)!}{|N|!}(\overline{v}_{\hat{\K}}(S \cup \{i\}) - \underline{v}_{\hat{\K}}(S)),
	\end{equation*}
	which is exactly $\phi_i(v_i)$. The proof of the second equality is merely a technicality left for the Appendix~\ref{app:mere-technicality}.
\end{proof}

Now, we shall demonstrate several key properties of the utopian gap, starting with its non-negativity.

\begin{proposition}
	\label{prop: gap non-negative}
	Let $(N, \hat{\K}, v)$ be an \S-extendable incomplete game. Then $\utopianGap_{(N,v)}(\K) \ge 0$. Moreover, $\utopianGap_{(N,v)}(\K)=0$ if and only if $\S(\hat{\K},v)$ is a singleton.
\end{proposition}
\begin{proof}
	Non-negativity follows from the fact that for $(N,\underline{v}) \in \S(\hat{\K},v)$ it holds that
	\begin{equation}\label{eq:gap-not-nega2}
		\sum_{i \in N}\phi_i(v_i) \ge \sum_{i \in N}\phi_i\left(\underline{v}_{\hat{\K}}\right) = v(N).
	\end{equation}
	The inequality in~\eqref{eq:gap-not-nega2} is strict if and only if there is more than one $\S$-extension as $\phi_i(v_i) > \phi_i(\underline{v}_{\hat{\K}})$ for at least one $i \in N$.
\end{proof}

Note that the set of extensions can be non-trivial even if we are missing any single value of the underlying game, see Appendix~\ref{app:zero-gap-not-attainable}.

\begin{proposition}
	\label{prop: gap in K}
	The cumulative utopian gap $\utopianGap$ is in the set $\K$
	\begin{enumerate}
		\item monotonically non-increasing, i.e., for any $S\in 2^N$
		      \begin{equation}\label{eq:gap-is-mono}
			      \utopianGap_{(N,v)}(\K) \ge \utopianGap_{(N,v)}(\K\cup \{S\}),
		      \end{equation}
		\item subadditive
		      , i.e., for $\K,\mathcal{L} \subseteq 2^N$, $\K \cap \mathcal{L} = \emptyset$,
		      \begin{equation}\label{eq:gap-is-sub}
			      \utopianGap_{(N,v)}(\K) + \utopianGap_{(N,v)}(\mathcal{L}) \geq \utopianGap_{(N,v)}(\K \cup \mathcal{L}).
		      \end{equation}
	\end{enumerate}
\end{proposition}
\begin{proof}

	To show~\eqref{eq:gap-is-mono}, for an incomplete game $(N, \hat{\K}, v)$ and a coalition $S\subseteq N$, consider the projection of $\S(\hat{\K}, v)\subseteq \mathbb{R}^{2^n}$ onto the axis, which is given by $S$. By Theorem~\ref{lem:super bounds} and supermodularity of the set of $\S$-extensions, such projection is $[\underline{v}_{\hat{\K}}(S), \overline{v}_{\hat{\K}}(S)]$. By Proposition~\ref{prop:gap-is-utopian} and the definition of $\phi_i(v_i)$, an $\varepsilon$ change in $\overline{v}_{\hat{\K}}(S)$ leads to
	\begin{equation*}
		\varepsilon \cdot \lvert S \rvert \cdot \frac{(\lvert S \rvert -1)!(\lvert N \rvert - \lvert S \rvert)!}{\lvert N \rvert !}
	\end{equation*}
	change in $\utopianGap$, and an $\varepsilon$ change in $\underline{v}(S)$ leads to
	\begin{equation*}
		-\varepsilon \cdot \lvert S \rvert \cdot \frac{\lvert S \rvert ! (\lvert N \rvert - \lvert S \rvert - 1)!}{\lvert N \rvert !}
	\end{equation*}
	change in $\utopianGap$.
	This means that $\utopianGap$ is increasing in $\overline{v}_{\hat{\K}}(S)$ and decreasing in $\underline{v}_{\hat{\K}}(S)$. As $\S(\hat{\K}\cup\{S\}, v)\subseteq \S(\hat{\K}, v)$, the projection of $\S(\hat{\K}\cup \{S\},v)$ on axis $S$ is contained in the projection $\S(\hat{\K}\cup \{S\},v)$ on axis $S$, which implies~\eqref{eq:gap-is-mono}.
	Finally, from non-negativity and monotonicity follows~\eqref{eq:gap-is-sub}.
\end{proof}

Finally, let us touch on the geometric interpretation of the utopian gap.
As stated in Theorem~\ref{lem:super bounds}, the upper and lower games define a hyper-rectangle which tightly bounds the set of extensions $\S(\hat{\K}, v)$.
Since the utopian gap is a linear function of the values of the lower/upper games, it can be seen as a \enquote{measure} of the size of $\S(\hat{\K}, v)$.
Uncertainty in each coalition is additionally weighted according to its contribution to the Shapley value.

\subsection{Principal's Optimization Problems}

A large utopian gap makes it difficult to distribute payoff, as each player tends to their utopian game. The more we know about the underlying game, the smaller the cumulative utopian gap gets, being zero for a game with a single $\S$-extension. However, to obtain all the necessary unknown values might be too expensive, since there are exponentially many of them.

To minimize the utopian gap, we hence assume the existence of a non-affiliated party who we refer to as the {\it principal}. Her task is to choose which coalitions should be investigated to reduce the utopian gap the most. We further assume the principal holds a level of expertise that guides the selection process. This expertise is formalized as a {\it prior distribution} over a set of potential characteristic functions. For example, in a medical context, a doctor acting as the principal might seek to assess a patient's response to a combination of drugs (represented as a coalition) and base their expectations on past clinical experience. Similarly, a machine learning engineer could rely on their prior knowledge of feature importance gained from previous problem-solving experiences. Consequently, we assume that each problem instance can be viewed as a sample drawn from this known prior distribution, denoted as $\valueDistribution$. To put bluntly, the principal is aware of the prior distribution, but not of the specific instance of the underlying game.

There are two basic approaches to choosing which coalitions to investigate, \textit{online} and \textit{offline}.  In the online approach, the principal operates sequentially, utilizing information from previously revealed coalition values.

\begin{definition}[Online Principal's Problem]\label{def: online principal problem}
	Let $t\in\mathbb{N}$, $\valueDistribution$, $\text{supp}\ \valueDistribution\subseteq\S$ be a distribution of superadditive games.
	Then $\K^*_t\subseteq 2^N \setminus \K_0$ is a solution of the \textit{online principal's problem} of size $t$ if
	\begin{equation}
		\label{eq: online principals problem}
		\K^*_t \in
		\argmin_{\K_t\subseteq 2^N\setminus \K_0, |\K_t|= t}\left\{
		\mathop{\mathbb{E}}_{v\sim \valueDistribution}\left[\utopianGap_{(N,v)}(\K_t)\right]\right\}
	\end{equation}
	where $\K_\tau$, $\tau \le t$ is such that $\K_\tau = \K_{\tau-1} \cup \{S_\tau\}$, $S_\tau = \pi(N,\K_{\tau-1},v)$, and $\pi$ is a policy function that chooses $S_\tau$ based on the known values of $v$, i.e. values of coalitions in $\K_{\tau-1}$.
\end{definition}
\noindent In contrast, the offline approach entails a lack of such information.

\begin{definition}[Offline Principal's Problem]
	Let $t\in\mathbb{N}$, $\valueDistribution$, $\text{supp}\ \valueDistribution\subseteq\S$ be a distribution of superadditive games.
	Then $\K^*_t\subseteq 2^N \setminus \K_0$ is a solution of the \emph{offline principal's problem} of size $t$ if
	\begin{equation}
		\label{eq: offline principals problem}
		\K^*_t \in
		\argmin_{\K\subseteq 2^N\setminus \K_0, |\K|= t}\left\{
		\mathop{\mathbb{E}}_{v\sim \valueDistribution}\left[\utopianGap_{(N,v)}(\K)\right]\right\}.
	\end{equation}
\end{definition}

\subsection{Algorithms Solving the Principal's Problems}

In this section, we discuss various methods for finding (approximate) solutions to the principal's problems defined above. We defer further details about all algorithms to Appendix~\ref{app: alg specs}.

\subsubsection{Offline Algorithms}

At each step $t$, the {\sc Offline Optimal} algorithm chooses coalitions $\{S_i\}_{i=1}^t$ which minimize the expected utopian gap under $\valueDistribution$. We estimate the expectation w.r.t. $\valueDistribution$ in Eq.~(\ref{eq: offline principals problem}) by $\kappa$ samples, see Algorithm~\ref{algo: offline optimal}.
\begin{algorithm}[h!]
	\caption{\sc Offline Optimal}
	\label{algo: offline optimal}
	\SetAlgoLined
	\DontPrintSemicolon
	\KwIn{distribution of superadditive functions $\valueDistribution$, number of steps $t$, number of samples $\kappa$}
	\vspace{1ex}
	$ \overline \K \gets 2^N \setminus \K_0 $\;
	$G \gets \left\{ \right\}$ \tcp*{trajectories and their average gap}
	\For(\tcp*[f]{for each trajectory}){$\mathcal{S} \subseteq \overline \K: \absolute{\mathcal{S}} = t$}
	{
		$\mu \gets 0$\;
		\For(\tcp*[f]{approximate expectation}){$j \in \{1, \dots, \kappa\}$}
		{
			$ v \sim \valueDistribution$\;
			$\mu \gets \mu + \utopianGap_{(N,v)}(\mathcal{S})$\;
		}
		$ \mu \gets \mu \slash \kappa $\;
		$ G \left[ \mathcal{S} \right] \gets \mu $
	}
	$ \left\{ S_i \right\}_{i=1}^t \gets \argmin_{\mathcal S \subseteq \overline \K: \absolute{\mathcal S} = t} G \left[ \mathcal S \right]$\;
	\Return{$\{S_i\}_{i=1}^{t}$}\;
\end{algorithm}

A computationally less demanding variant of the {\sc Offline Optimal} is the {\sc Offline Greedy} algorithm.
It chooses the next coalition $S_t$ such that, given the previous trajectory $\{S_i\}_{i=1}^{t-1}$, it minimizes the expected utopian gap.
Consequently, it can perform no better than the {\sc Offline Optimal}.
We again estimate the expectation in Eq.~(\ref{eq: offline principals problem}) by $\kappa$ samples, see Algorithm~\ref{algo: offline greedy}.
\begin{algorithm}[ht]
	\caption{\sc Offline Greedy}
	\label{algo: offline greedy}
	\SetAlgoLined
	\DontPrintSemicolon
	\KwIn{distribution of superadditive functions $\valueDistribution$, number of steps $t$, number of samples $\kappa$}
	\vspace{1ex}
	\SetKwFunction{FNC}{Offline Greedy}
	\SetKwProg{Fn}{Function}{:}{}
	$\{S_i\}_{i=1}^{t - 1} \gets $ {\FNC{$\valueDistribution, t-1$}}\;
	$ \overline \K \gets 2^N \setminus ( \K_0 \cup \left\{ S_i \right\}_{i=1}^{t-1} ) $\;
	$G \gets \left\{ \right\}$ \tcp*{trajectories \& their average gap}
	\For(\tcp*[f]{for each trajectory}){$S \in \overline \K$}
	{
		$\mu \gets 0$\;
		\For(\tcp*[f]{approximate expectation}){$j \in \{1, \dots, \kappa\}$}
		{
			$ v \sim \valueDistribution$\;
			$\mu \gets \mu + \utopianGap_{(N,v)}(\left\{ S_i \right\}_{i=1}^{t-1} \cup \left\{ S \right\})$\;
		}
		$ \mu \gets \mu \slash \kappa $\;
		$ G \left[ S \right] \gets \mu $
	}
	$S_{t} \gets \argmin_{S\in \overline \K} G \left[ S \right]$\;
	\Return{$\{S_i\}_{i=1}^{t}$}\;
\end{algorithm}

\subsubsection{Online Algorithm}

In comparison to the offline problem, solving the online problem poses a significantly greater challenge. Intuitively, one key reason is that an algorithm for the online problem must compute (or approximate) a restriction of $\valueDistribution$ that remains consistent with the values it has uncovered in prior steps. However, this can be particularly challenging, especially in case the only access to $\valueDistribution$ is through sampling. To tackle the online problem and derive an approximate solution, we employ reinforcement learning~\cite{sutton2014reinforcement}, specifically, the proximal policy optimization~({\sc PPO})~\cite{Schulman2017}.

At each step $\tau$, {\sc PPO} receives values of coalitions $\K_{\tau-1}=\{S_i\}_{i=1}^{\tau-1}$ it uncovered in the past and chooses the next coalition $S_\tau$.
To get a strategy which efficiently minimizes the utopian gap, we define the reward (which is maximized by the {\sc PPO} algorithm) as the negative expected utopian gap averaged over $\tau \le t$.

\medskip

As previously mentioned, the greedy algorithm is significantly more computationally efficient, with a linear complexity in the number of coalitions, while the optimal variant exhibits an exponential time complexity. A question which naturally arises is under which conditions the local greedy search can reliably yield a globally optimal solution. While our empirical observations indicate similar performance between both approaches, the greedy method is not optimal in general. To illustrate, we provide a simple example in Appendix~\ref{app: local is not global example}, where the local search fails to identify the optimum. However, it is known that, if the optimized function is supermodular, the locally optimal steps are guaranteed to yield a $(1-1/e)$-approximation of the global optimum~\cite[Proposition 3.4]{Nemhauser1978}.
In our case, the utopian gap is supermodular for all small games.

\begin{figure*}[t!]
	\includegraphics[width=0.48\textwidth]{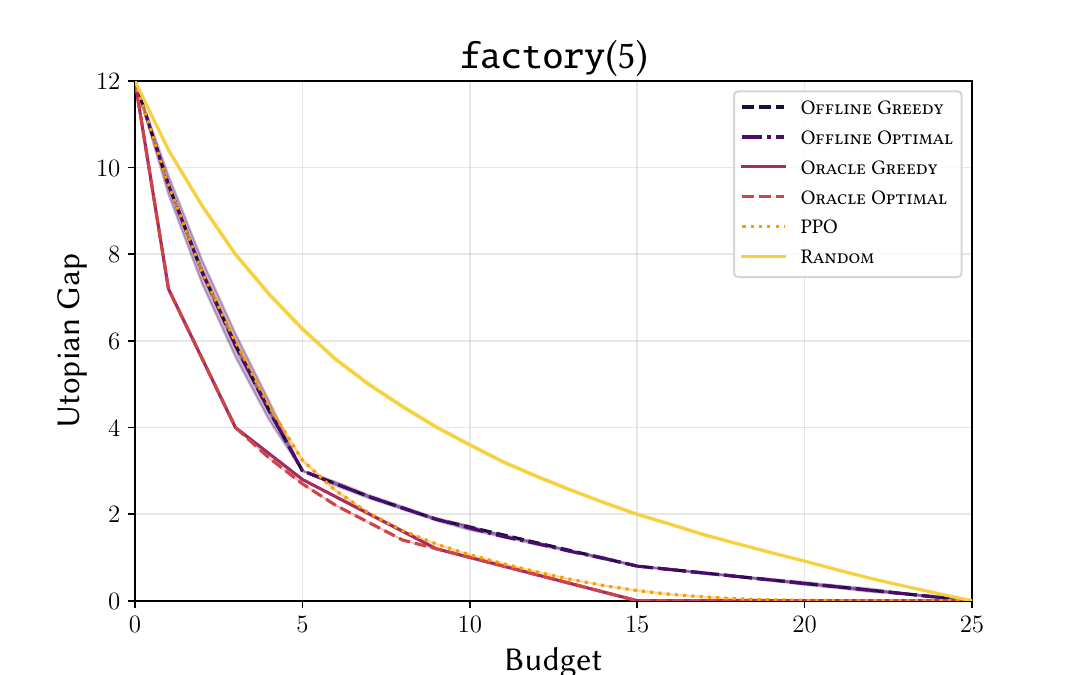}
	\includegraphics[width=0.48\textwidth]{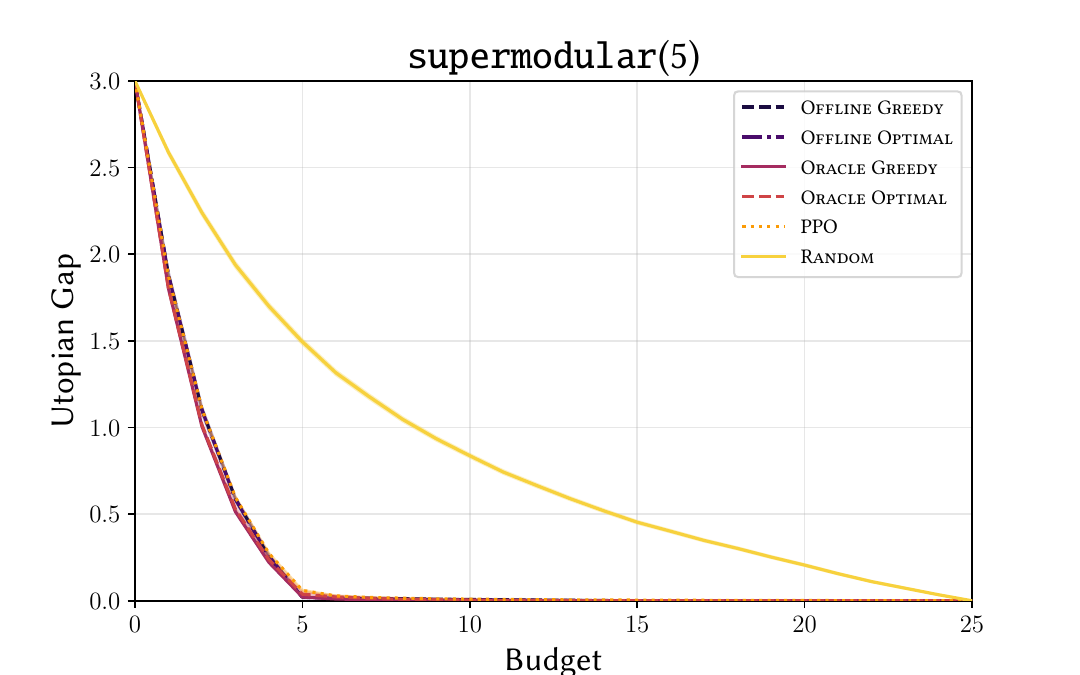}
	\caption{
		The utopian gap as a function of number of revealed coalitions (i.e. steps of the Principals problem) for different algorithms.
		We show {\tt factory}(5) (left), and {\tt supermodular}(5) (right) games.
		All algorithms outperform the {\sc Random} benchmark considerably.
		The greedy versions of each algorithm exhibit similar performance to the optimal variants.
		The {\sc PPO} algorithm is initially close to the offline algorithms, and uses the online information to approach the oracle algorithms.
	}
	\label{fig: exploitability over time}
\end{figure*}

\begin{proposition}\label{prop:gap-supermodularity}
	For $|N|\le 4$, the gap of an incomplete $\S$-extendable game is supermodular.
\end{proposition}
\begin{proof}[Proof Sketch]
or $\lvert N \rvert =3$, the gap is modular due to invariant bounds post-revelation. For $|N|=4$, the supermodularity is demonstrated through a more technical case analysis.
The full proof can be found in Appendix~\ref{app: supermodularity for n<5}.
\end{proof}
However, already for $|N|=5$, it can be shown the gap is not supermodular for several significant subsets of superadditive games, such as totally monotonic games (a subset of convex games, which under the name \emph{belief measures} forms the foundation on the \emph{evidence theory}~\cite{Grabisch2016, shafer1976}), symmetric games (often studied for their robust properties with respect to solution concepts~\cite{Peleg2007}), or graph games (games with a compact structure represented by $O(n^2)$ values with application in scenarios where only bilateral collaboration between players occurs~\cite{Chalkiadakis2012}). 

\begin{observation}\label{prop:examples-of-games}
For $|N|=5$, there exist a totally monotonic game, convex game, symmetric superadditive game, and graph game for which the utopian gap is not supermodular.
\end{observation}

For even larger games, we can formualate a precise criterion to ascertain the non-supermodularity of the utopian gap.

\begin{proposition}\label{prop:gap-criterion}
For $|N|\geq 6$, the gap of an incomplete $\S$-extendable game is not supermodular if there exist $i,j,k,l \in N$ such that $\tilde{v}(ij) \le \tilde{v}(jk) \le \tilde{v}(kl)$, where $\tilde{v}(S) = v(S) - \sum_{i\in S}v(\{i\})$, and
    \begin{equation}\label{gap-criterion}
        \tilde{v}(kl) < \left[{\binom{n}{2}}{\binom{n}{\lfloor n/2 \rfloor}}^{-1}\left(2^{n-3}-n+2\right)-1\right]\tilde{v}(ij).
    \end{equation}
\end{proposition}
\begin{proof}[Proof Sketch]
We rephrase the supermodularity condition~\eqref{eq: supermodular constraints} in relation to the utopian gap, and demonstrate that this condition is never met when Eq.~\eqref{gap-criterion} holds for the pairs in some $\{i,j,k,l\}$.
The full proof can be found in Appendix~\ref{app:gap-criterion}.
\end{proof}

This criterion allows to construct examples also for all earlier mentioned classes of games. The proofs and examples are available in Appendix~\ref{app:counter-examples}. Despite these findings, empirically, the greedy algorithm demonstrates similar performance to its optimal counterpart, as illustrated in the following section. 

\section{Empirical Evaluation}\label{sec: experiments}

Finally, we demonstrate the performance of our algorithms on practical examples. We outline the domains used for the evaluation, introduce baseline methods for comparison, detail the algorithmic setups, and conduct a comprehensive analysis of the gathered results. All the details that could not be accommodated within the main text have been addressed in the relevant appendices.

\subsection{Experimental Domains}

We conduct our experiments on two representative families of cooperative games. Their main difference is the degree to which the values of different coalitions are correlated.

For the tightly correlated scenario, we use a simple model of a factory.
The set of players $N$, $|N|\ge 2$ consists of a single factory owner $o$ and $n-1$ identical workers.
The value of a coalition $S$ is
\begin{equation}
	v(S) =
	\begin{cases}
		|S|-1 & \text{\ if\ }o\in S, \\
		0     & \text{\ otherwise}.
	\end{cases}
	\label{eq: factory definition}
\end{equation}
In words, each worker contributes equally to the value of a coalition as long as he has somewhere to work.
It is easy to verify that the resulting game is superadditive. We denote the distribution of factory games of $n$ players where the owner is selected uniformly at random as {\tt factory}$(n)$.
We show that the gap is not supermodular already for {\tt factory}($5$) in Appendix~\ref{app: local is not global example}.

The second family is the broad class of supermodular games. These games are among the most studied in the cooperative game theory~\cite{Grabisch2016} -- some authors consider this property so important they even impose it in the definition of a cooperative game~\cite{owen2013}. In our context, the problem of identifying a subset of coalitions which represents a supermodular function well has been studied in a broader context of set functions~\cite{Vondrak2014}. In~\cite{Beliakov2022}, the authors derived an efficient algorithm for uniform sampling of monotone supermodular functions. We refer to the distribution of corresponding games of $n$ players as {\tt supermodular}$(n)$.

We conduct further experiments on several other classes of cooperative games as well, and present these results in Appendix~\ref{app: additional experiments}. 
We observe similar trends as those presented here in the main text.

\subsection{Benchmarks}

We compare our algorithms introduced in the previous section to three baselines: a random algorithm and two oracle algorithms. The random algorithm selects the next coalition uniformly at random. We refer to it in the results as {\sc Random}. Next, we introduce two oracle methods that are not deployable in practice because they assume the knowledge of the underlying true characteristic function. However, they provide an upper bound on what an optimal online algorithm could possibly achieve.

The {\sc Oracle Optimal} algorithm operates similarly to {\sc Offline Optimal}, with the key difference being that it leverages an oracle to acquire values of the underlying game $(N, v)$ before making the next coalition selection. Consequently, {\sc Oracle Optimal} can harness complete knowledge of the underlying game to minimize the gap effectively. Its pseudocode as Algorithm~\ref{algo: online optimal} is in Appendix~\ref{app:oracle-algs}.

\begin{figure*}[t!]
	\vspace{-5ex}
	\includegraphics[width=\textwidth]{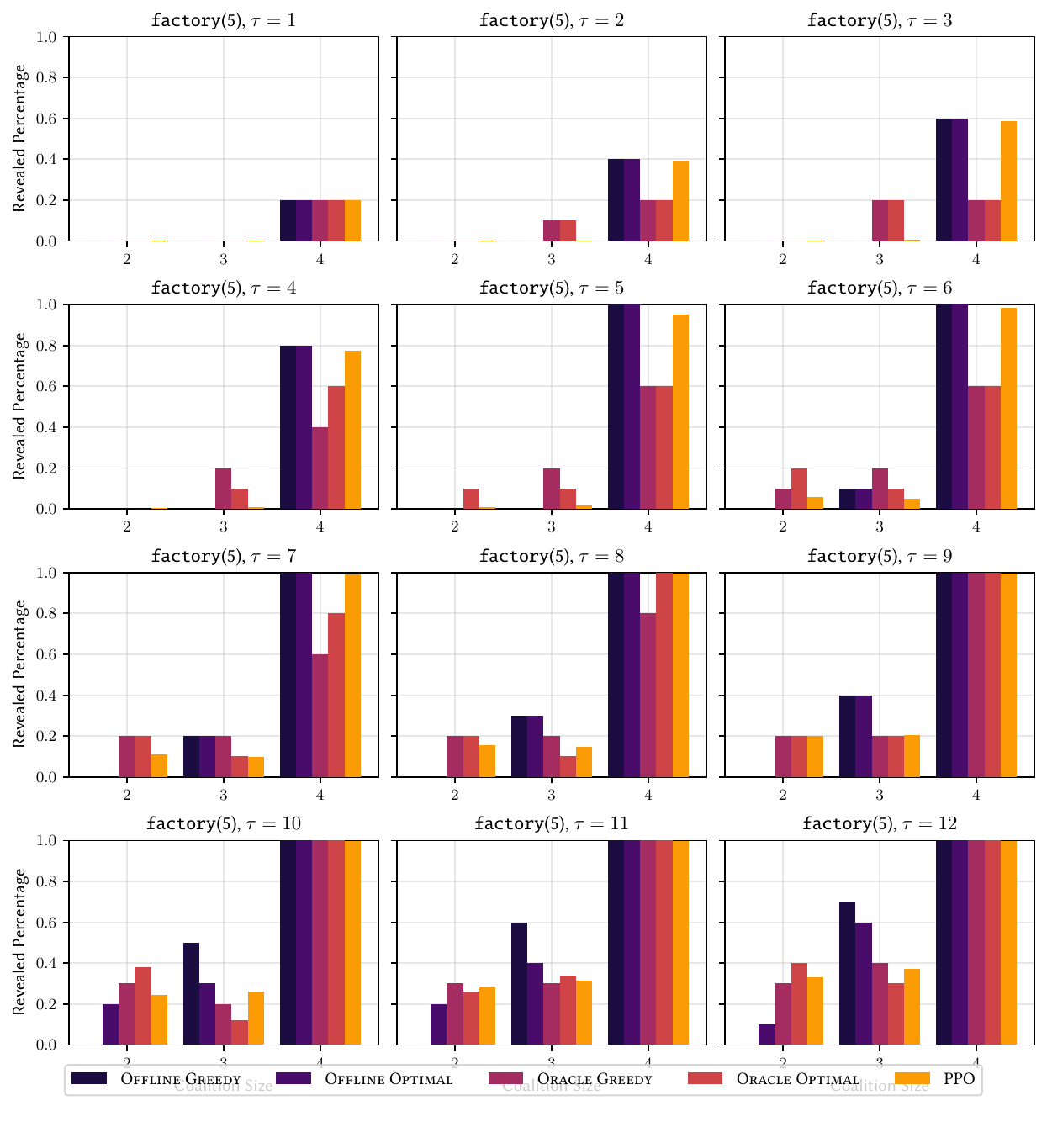}
	\caption{
		Percentage of coalitions of the same size selected up to step twelve for {\tt factory}(5) and each algorithm.
		Results show clear preference for larger coalitions, i.e. they contribute more information about the cooperative game on average.
		The oracle algorithms favor smaller coalitions earlier, suggesting the representation of a specific game can efficiently use even smaller coalitions.
			{\sc PPO} initially behaves similarly to the offline algorithms.
		At later steps, it uses the previously obtained values and its selections resemble the oracle methods.
		See Figure~\ref{fig: app factory5 bar plots} in Appendix~\ref{app: additional experiments} for a plot showing individual coalitions.
	}
	\label{fig: factory5 bar cum plots}
\end{figure*}

Similar to the {\sc Offline Greedy}, the {\sc Oracle Greedy} algorithm selects next coalition $S_t$ such that, in combination with the previous trajectory $\{S_i\}_{i=1}^{t-1}$, it minimizes the utopian gap.
It also uses the oracle to gather all information about the underlying game $(N, v)$.
Its pseudocode is depicted as Algorithm~\ref{algo: online greedy} in Appendix~\ref{app:oracle-algs}.

\begin{figure}[t!]
	\includegraphics[width=0.48\textwidth]{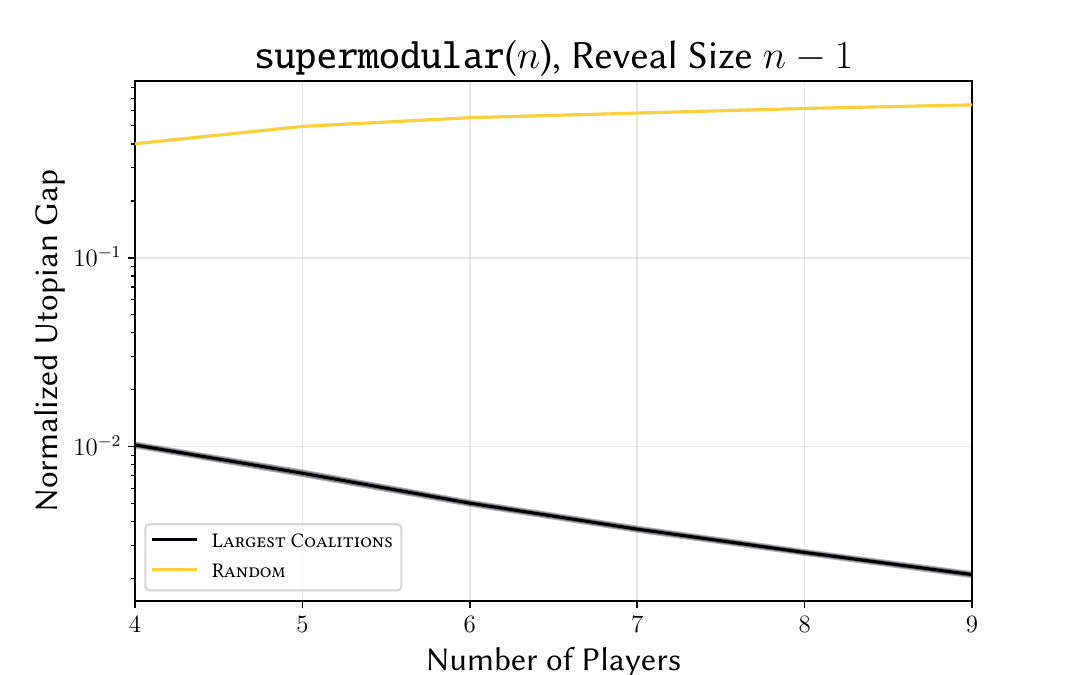}
	\caption{
		The utopian gap as a function of the number of players.
		The figure compares expected gap of {\tt supermodular}($n$) when choosing coalitions randomly, and when all coalitions of size $n-1$ are selected.
		Not only is the utopian gap in the latter case small, it decreases with the number of players.
	}
	\label{fig: reveal larges coalitions}
\end{figure}

\subsection{Experimental Setup}

During the training phase, we execute the {\sc PPO} algorithm for a total of \ppoTimesteps time-steps, which, in the case of 12-step instances, translates to \ppoIterations iterations. In each iteration, we repeatedly collect \ppoTrajectories trajectories by sampling batches from the distribution $\valueDistribution$. Following each iteration, we optimize the {\sc PPO} surrogate objective for \ppoEpochs epochs. To ensure uniformity, we normalize the input values $v(S)$ to a unit interval, as elaborated in Appendix~\ref{app: value normalization}. For the two optimal algorithms, we employ $\kappa=\numSamplesFive$ samples to estimate the mean value in Eq.~\eqref{eq: offline principals problem}.
See Appendix~\ref{app: alg specs} for more details.

Among the algorithms we have proposed, both {\sc Offline Optimal} and {\sc Oracle Optimal} come with significantly higher computational costs (their complexity scales as $\mathcal{O}(2^{2^n})$) when compared to their more computationally efficient greedy counterparts (which scale $\mathcal{O}({2^n})$). To provide some perspective, our evaluation for $n=5$ with twelve steps and ten samples necessitated approximately \hoursFiveOptimal hours to finish. Extrapolating, we estimate that completing the evaluations for the remaining steps in this setup, across all ten samples, would require roughly \hoursFiveOptimalAll hours. Furthermore, extending these experiments to $n=6$ would demand \timeSixOptimal.

\paragraph{Hardware and Software}
All experiments are conducted on a computational cluster with AMD~EPYC~7532~CPUs running at 2.4~GHz.
When running algorithms on five players games, we utilize 15~cores and 12~GB of RAM.
The code was implemented in Python~3.10 using pytorch~2.0, stable\_baselines3~2.0, and gymnasium~0.28.
Additional details are presented in Appendix~\ref{app: alg specs}, and is available at GitHub~\cite{gitrepo}.

\subsection{Results}
\label{ssec: factory experiments}

First, we study the {\tt factory}($5$) distribution.
Figure~\ref{fig: exploitability over time} shows the dependence of the utopian gap on the number of revealed coalitions. As expected, the oracle algorithms outperform their offline variants, especially at the beginning. {\sc PPO} is initially similar to the offline algorithms, and later uses the online information to reach similar values of the utopian gap as the oracle methods. The greedy algorithms exhibit similar performance to their optimal counterparts.

To further illustrate which coalitions are important, Figure~\ref{fig: factory5 bar cum plots} shows the percentage of coalitions of a given size selected by each algorithm. The offline algorithms favor the largest coalitions, suggesting they carry the most information overall. The oracle algorithms can gain small improvements by selecting smaller coalitions. This signals that, for a specific game, a tailored representation including also smaller coalitions is better.

Next, we focus on the {\tt supermodular}($5$) distribution.
We show the evolution of the utopian gap as a function of the number of revealed coalitions in Figure~\ref{fig: exploitability over time}.
Surprisingly, the utopian gap can be minimized very fast, by about $99\%$ by step five.
We observed similar trends for {\tt supermodular}(4) presented in Appendix~\ref{app: factory additional results}.
In each case, the algorithms prefered coalitions of size $n-1$, suggesting they are the most important among the unknown coalitions.
We further evaluate the impact of the largest coalitions, by computing the utopian gap as a function of $n$ when all coalitions of size $n-1$ are revealed.
We compare this simple heuristic with {\sc Random} in Figure~\ref{fig: reveal larges coalitions}.
The utopian gap is consistently very low and, surprisingly, decreases with the number of players.
This suggest that most of the information about a supermodular game can be captured by values of $\mathcal{O}(n)$ coalitions.

\section{Conclusion}
\label{sec: cenclusion}

In this paper, we study strategies for efficiently mitigating uncertainty within incomplete cooperative games. We introduce the concept of the \enquote{utopian gap}, which quantifies the disparity between players' expectations and a realistic game outcome. We show fundamental properties of the utopian gap that enables us to compute it more efficiently. Geometrically, the utopian gap reflects the number of potential extensions of an incomplete game. We focus on reducing the utopian gap through well-informed queries about the unknown values within the incomplete game, effectively constructing a representation tailored to capture a maximum information about the game in both online and offline fashion. Our findings indicate that our approach significantly outperforms random queries and approaches optimality. Particularly noteworthy is our heuristic for supermodular games, which reduces the utopian gap by orders of magnitude while requiring only $\mathcal{O}(n)$ queries -- an amount logarithmic in the number of coalitions.

\paragraph{Future Work}

We posit that our approach is not confined solely to the boundaries of cooperative game theory and can be extended to more general set functions. This extension assumes the existence of criteria involving functions (in this work, it is the Shapley value).
that map the powerset of a set $N$ to subspaces of dimensions linear in the size of $N$. Still in the context of games, we see applications of our approach within the SHAP~\cite{NIPS2017_7062} framework, as the original formulation can only handle comparatively small models. Finally, we would also like to strengthen our result by providing guarantees by working within the regret minimization framework.

\clearpage
\begin{acks}
The authors would like to thank Martin Loebl and Milan Hlad\'{i}k for their insightful comments. David Sychrovsk\'{y} and Martin \v{C}ern\'{y} received support from the Charles University Grant Agency (GAUK 206523). This work was supported by the Czech Science Foundation grant no. P403-22-11117S, the CoSP Project grant no. 823748, and the Charles University project UNCE 24/SCI/008. 
Computational resources were supplied by the project e-Infrastruktura CZ (e-INFRA LM2018140) provided within the program Projects of Large Research, Development and Innovations Infrastructures.
\end{acks}

\bibliographystyle{ACM-Reference-Format}
\bibliography{main}

\clearpage
\appendix
\section{All values are needed for zero utopian gap in general}\label{app:zero-gap-not-attainable}
We remark that throughout the paper, we do not aim to achieve a zero gap, as this is impossible for some games until all of the coalition values are known. To see this, consider strictly superadditve incomplete game $(N,\hat{\K},v)$, for which a single coalition value is unknown. Formally, let $(N,\hat{\K},v)$ with $\hat{\K} = 2^N \setminus \{S^*\}$ where $S^* \in 2^N \setminus\K_0$, be such that
\[
	v(S) + v(T) < v(S \cup T)
\]
for all $S,T \subseteq N, S \cap T = \emptyset$. Then as $v(S^*) < v(S^* \cup T) - v(T)$ for every $T \subseteq N \setminus S^*$ and $v(X) + v(S^* \setminus X) < v(S^*)$ for every $X \subseteq S^*$, it follows from the definition of the lower/upper games that
\[
	\underline{v}_{\hat{\K}}(S^*) < v(S^*) < \overline{v}_{\hat{\K}}(S^*).
\]
This means there is more than one $\S$-extension, which by Proposition~\ref{prop: gap non-negative} yields $\utopianGap_{(N,v)}(\K) > 0$.

\section{Algorithm Specifications}
\label{app: alg specs}

\paragraph{Oracle and Offline Algorithms}
Greedy strategies are computationally straightforward compared to their optimal counterparts.
A single step using a greedy approach demands $ \mathcal{O}(g \cdot 2^n) $ time for a single sample.
Here, $ g $ represents the time required to compute the utopian gap.

On the other hand, optimal algorithms prove to be computationally intensive due to the necessity of examining every sequence of actions, that is, every subset of $ 2^N $.
For each sample, the time complexity amounts to $ \mathcal{O}(g \cdot 2^{2^n}) $.
Despite our efforts to parallelize the computation where possible, given our available resources, we were only able to compute the complete optimal strategies for up to 5 players.
The computation for 5 players took roughly \hoursFiveOptimal hours.
We estimate, that to compute the optimal strategies for 6 players, for all steps, would take\timeSixOptimal.

The online variants are easier to parallelize, since each sample can be computed and evaluated independently, in contrast to the offline variants, where all the samples need to be computed, put together, averaged, and then finally evaluated.

When estimating the expectation with respect to $\valueDistribution$, we used \numSamplesStd samples for 4 players. For 5 players, we ended up using only \numSamplesFive samples.

\paragraph{Reinforcement Leaning}
We apply reinforcement leaning~\cite{sutton2014reinforcement} to approximate the optimal strategy of the online principal's problem.
Namely, we use the Proximal policy optimization~({\sc PPO})~\cite{Schulman2017}.
We want to find a strategy of the principal which efficiently minimizes the average cumulative utopian gap.
As such, we train {\sc PPO} to minimize the utopian gap
\begin{equation*}
	r(\K_{\tau-1}, S_\tau) = -\utopianGap_{(N,v)}(\K_{\tau-1}\cup \{S_\tau\}).
\end{equation*}
at every step, which provides a stronger learning signal compared to the final reward.
This is equivalent to training over a distribution of the online principal's problems with uniformly distributed size $t$.

In our implementation, we parametrize both actor and critic of the {\sc PPO} algorithm with a two-layer fully-connected neural network with 64 hidden units and ReLU activation each.
To optimize the surrogate {\sc PPO} objective, we used the Adam optimizer.
The rest of the hyperparameteres can be found in Table~\ref{tab: hyperparameters}.

\begin{table}[t]
	\begin{tabular}{c|c|c}
		\hline
		Parameter  & Value             & Description                    \\
		\hline
		\hline
		$\alpha_a$ & $3\cdot10^{-4}$   & Actor learning rate            \\
		$\alpha_c$ & $1.5\cdot10^{-4}$ & Critic learning rate           \\
		$\beta$    & $0.1$             & Entropy regularization         \\
		$\gamma$   & 1                 & Reward discounting rate        \\
		$\lambda$  & 0.95              & Generalized advantage estimate \\
		$\epsilon$ & 0.2               & Surrogate clip range           \\
		$B$        & $5\cdot 10^{4}$   & Rollout buffer size            \\
		$M$        & 0.5               & Max gradient norm              \\
		$n_e$      & 10                & Number of training epochs      \\
		\hline\end{tabular}
	\caption{Hyperparameters used during training.}
	\label{tab: hyperparameters}
\end{table}

\paragraph{Random Algorithm}
The {\sc Random} algorithm is computationally simple, requiring only $\mathcal{O}(g)$ time to compute a single sample.
As such, it took only a few minutes to compute for 5 players.
We again used \numSamplesStd samples to approximate the expectation and the standard deviation, for both 4 and 5 players.

\section{Example: Local Optimum is not Global Optimum}
\label{app: local is not global example}

\begin{figure}
	\centering
	\includegraphics[width=\columnwidth]{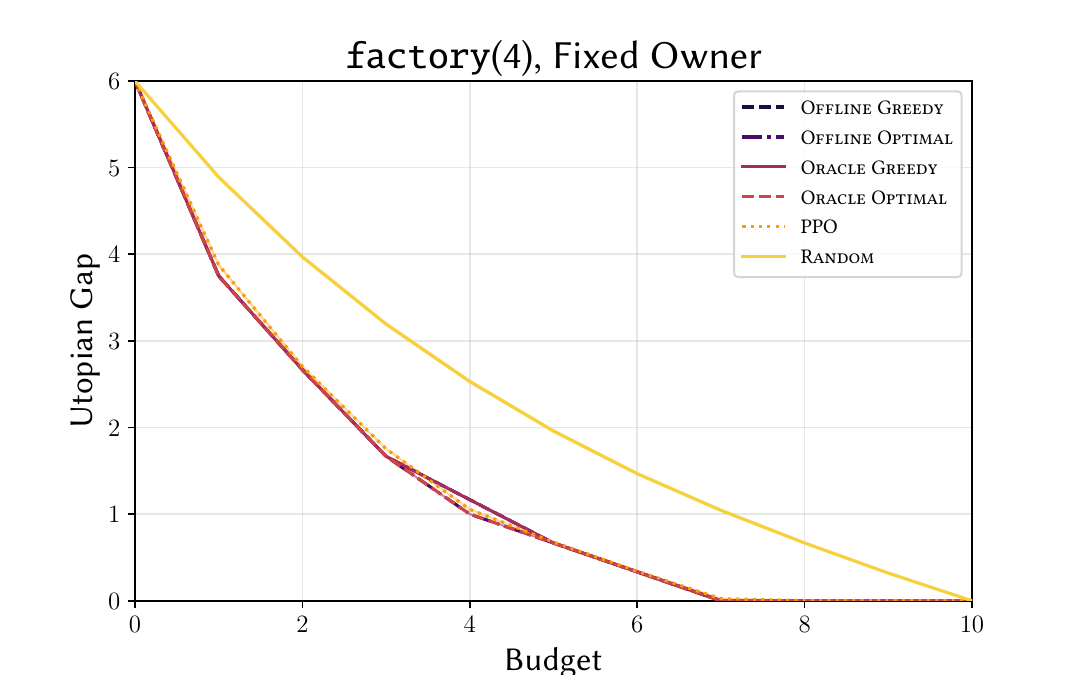}
	\caption{The cumulative utopian gap as a function of the number of revealed coalitions on a {\tt factory}(4) where the owner is fixed. The greedy algorithms fail to find the global optimum at step four, see Appendix~\ref{app: local is not global example}. The performance of oracle and offline algorithms is the same in this case, because $\valueDistribution$ includes just a single game. Finally, {\sc PPO} finds a strategy which is close to optimal, not greedy.}
	\label{fig: app fixed factory}
\end{figure}

In this section, we present an example of a sequence of locally optimal steps on the principal's problem which does {\it not} leading to a global optimum.
We will use a single {\tt factory}(4) game, i.e. the position of the owner is not randomized over.
We will refer to the owner as $o$, and the workers as $w_1, w_2, w_3$.
Note that since the game is fixed, {\sc Offline Optimal} has the same information as {\sc Oracle Optimal} and thus behaves the same.
Similar argument can be used to equate {\sc Offline Greedy} and {\sc Oracle Greedy}.
Thus, we only analyze the {\sc Oracle Optimal} and {\sc Oracle Greedy} algorithms.
The cumulative utopian gap as a function of the number of revealed coalitions is given in Figure~\ref{fig: app fixed factory}.

The {\sc Oracle Greedy} algorithm selects the coalitions $ \left\{ w_1, w_2, w_3 \right\} $, $ \left\{ o, w_1, w_2 \right\} $, and $ \left\{ o, w_3 \right\} $ at the first three steps, which matches the {\sc Oracle Optimal}.
The resulting final utopian gap is $\frac 59$ for both algorithms.
At step four, the {\sc Offline Greedy} chooses $\left\{ o, w_1, w_2 \right\}$.
However, because it chose $\left\{ o, w_3 \right\}$ as step three, it now doesn't match the trajectory of the {\sc Oracle Optimal}, which is to choose all the coalitions of size three.
As a result, the {\sc Oracle Greedy} strategy achieves utopian gap $\frac 7{18}$, while the {\sc Oracle Optimal} strategy achieves $\frac 13$.

Importantly, the {\sc Oracle Greedy} strategy strictly prefers its choice at $t=3$, resulting in $\frac{5}{9}$.
Choosing another one of the bigger coalitions, would result in utopian gap $\frac{7}{12}$.

Finally, the {\sc PPO} algorithm outperforms {\sc Oracle Greedy}.
As the latter bounds performance of all greedy online algorithms, it shows {\sc PPO} is not a greedy algorithm\footnote{This is because {\sc Oracle Greedy} outperforms all greedy algorithms.}.
This is consistent with it optimizing for the expected average utopian gap, see Appendix~\ref{app: alg specs}.

\subsection{Non-Supermodularity for $|N|>4$}
In order for the utopian gap to be supermodular, it needs to satisfy $\forall S,Z \in 2^N$, and $\forall \K \subseteq 2^N, \K_0 \subseteq \K$
\[
	\utopianGap_\Delta = \utopianGap_{(N,v)}(\K \cup S \cup Z) - \utopianGap_{(N,v)}(\K \cup S) - \utopianGap_{(N,v)}(\K \cup Z) + \utopianGap_{(N,v)}(\K),
\]
However, when the owner is player 1 and one selects
$\K = \K_0 \cup \{\{1,2,3\},\{1,4\} \}$, $S = \{1,2\}$, and $Z = \{1,2,3,5\}$,
the condition reads $ 7.3 - 7.9 - 8.1 + 8.6 = -0.1 < 0 $, which shows {\tt factory} ($5$) is not supermodular.
Similar argument can be made for $n>5$.

\section{Additional Experimental Results}
\label{app: additional experiments}

\subsection{Factory and Supermodular}
\label{app: factory additional results}
\begin{figure*}[t]
	\includegraphics[width=0.48\textwidth]{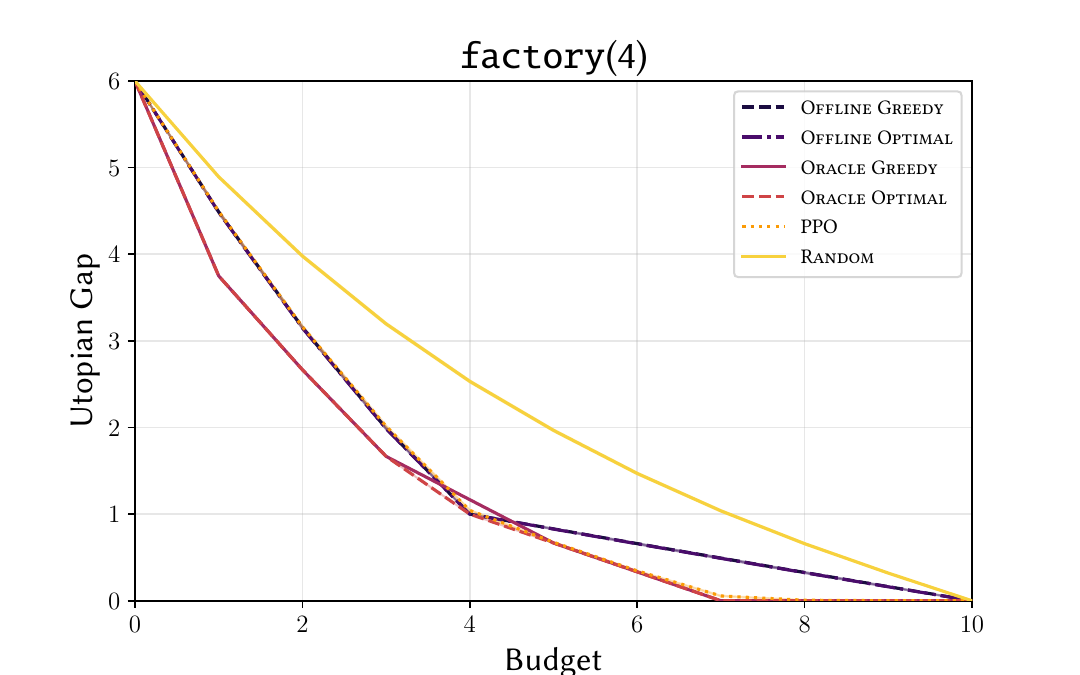}
	\includegraphics[width=0.48\textwidth]{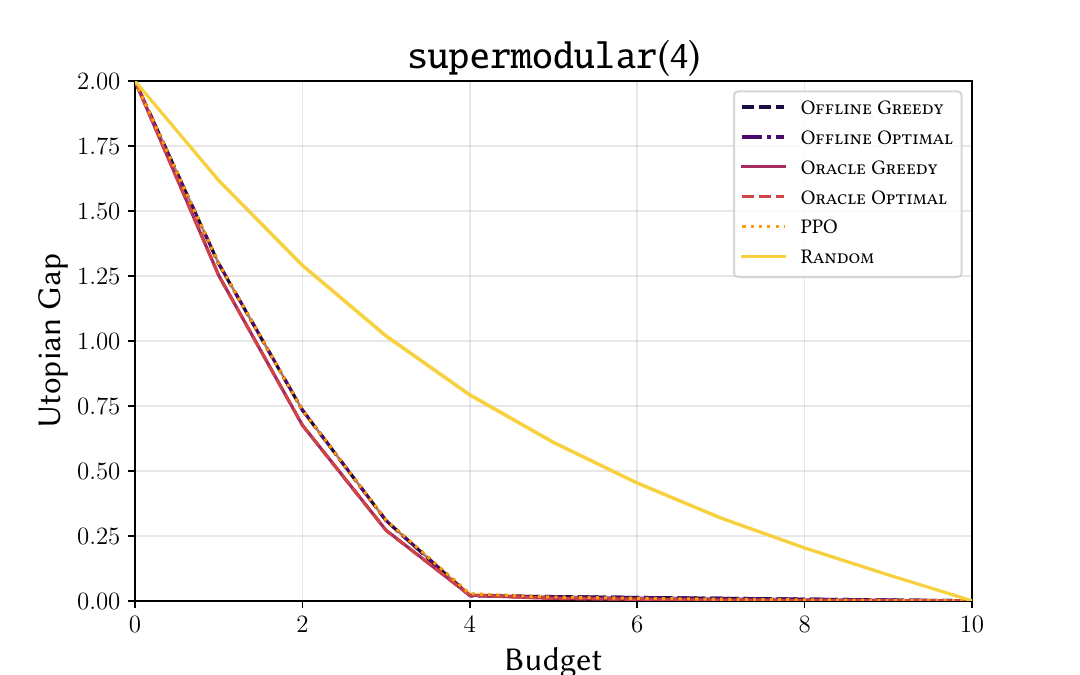}
	\caption{
		The utopian gap as a function of number of revealed coalitions (i.e. steps of the Principals problem) for different algorithms.
		We show {\tt factory}(4) (left), and {\tt supermodular}(4) (right) games.
		All algorithms outperform the {\sc Random} benchmark considerably.
		The greedy versions of each algorithm exhibit similar performance to the optimal variants.
		The {\sc PPO} algorithm is initially close to the offline algorithms, and uses the online information to approach the oracle algorithms.
	}
	\label{fig: app exploitability over time}
\end{figure*}
\begin{figure*}[t]
	\includegraphics[width=0.48\textwidth]{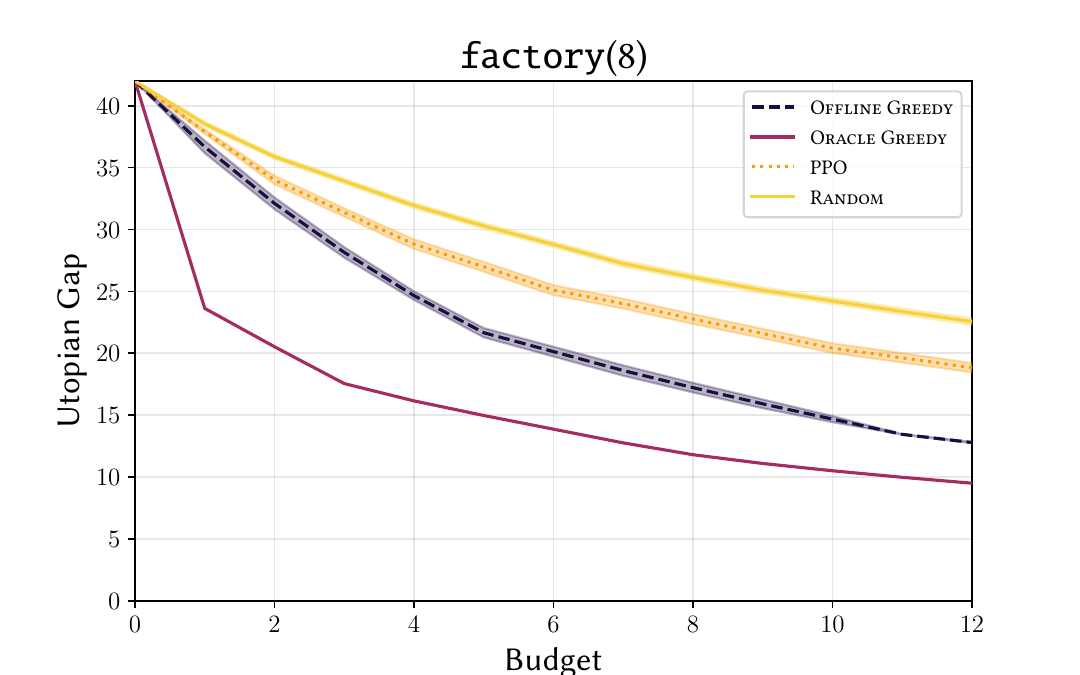}
	\includegraphics[width=0.48\textwidth]{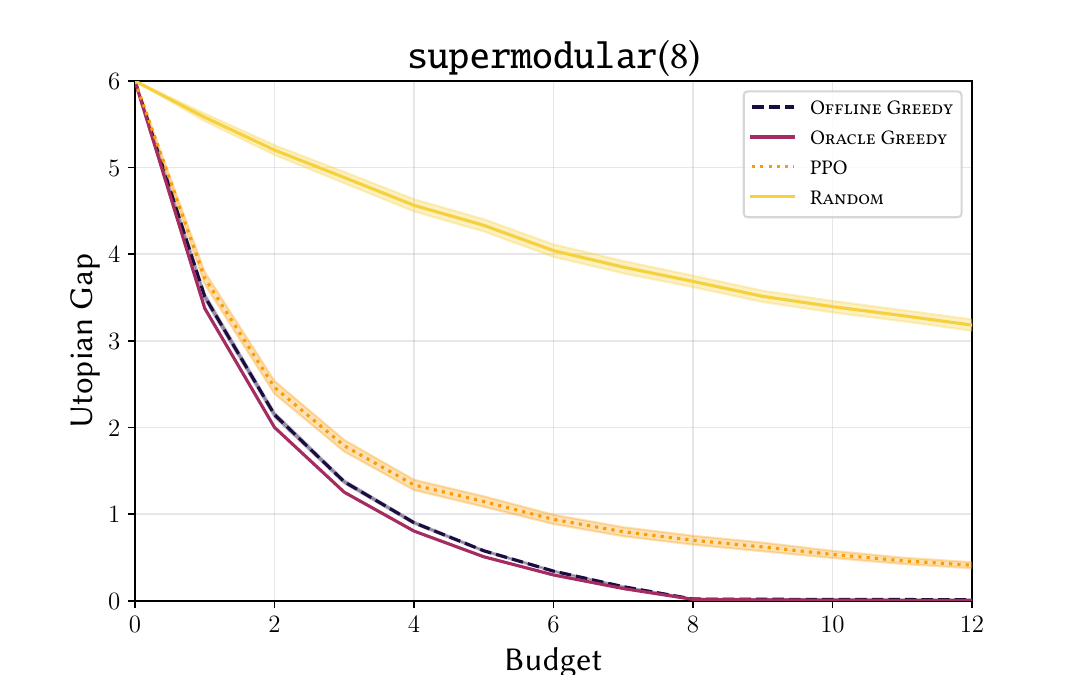}
	\includegraphics[width=0.48\textwidth]{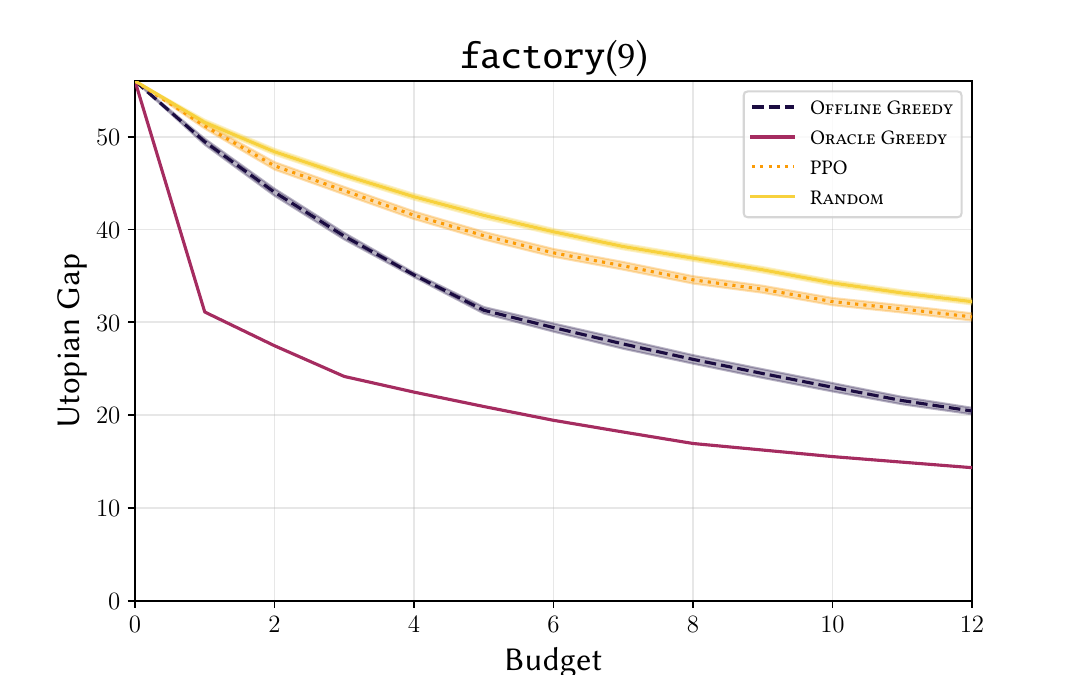}
	\includegraphics[width=0.48\textwidth]{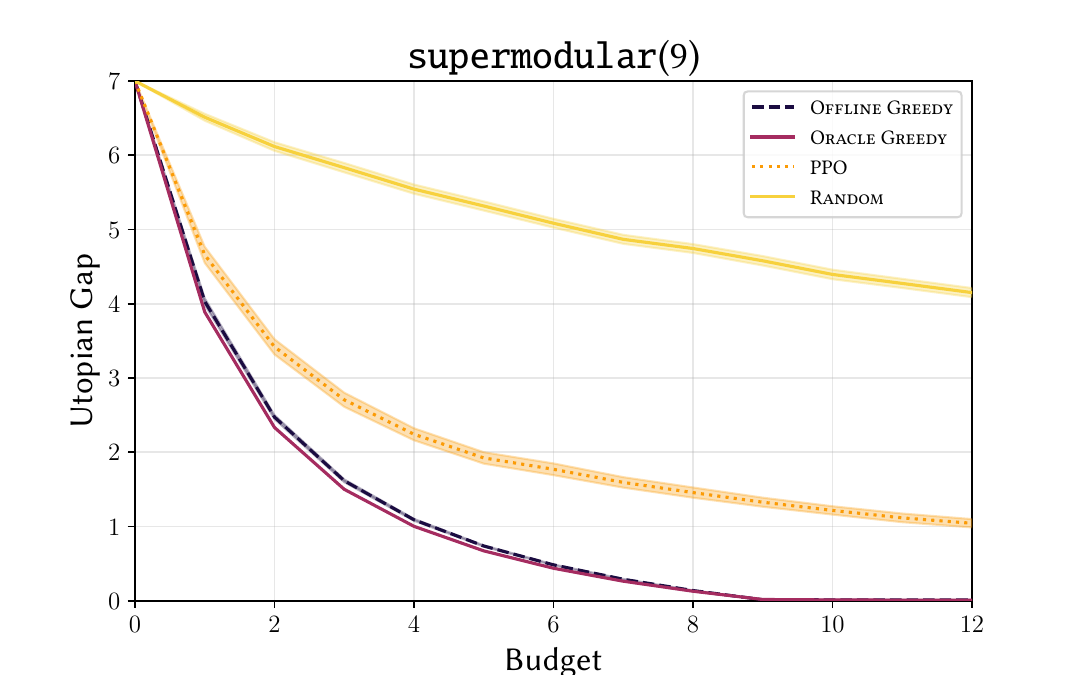}
	\includegraphics[width=0.48\textwidth]{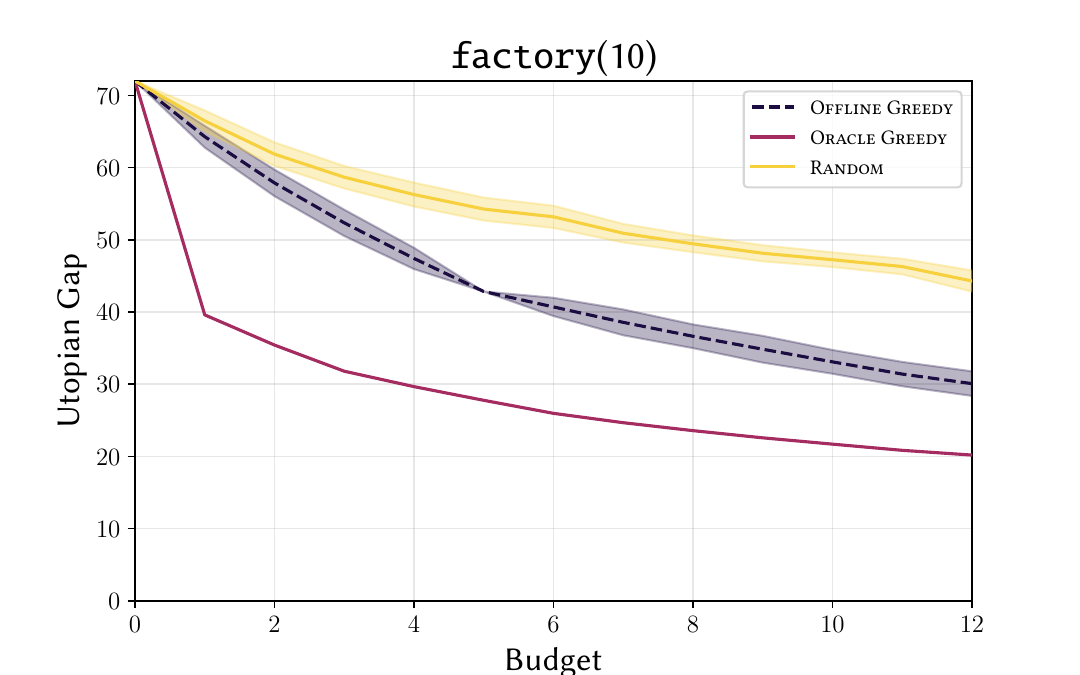}
	\includegraphics[width=0.48\textwidth]{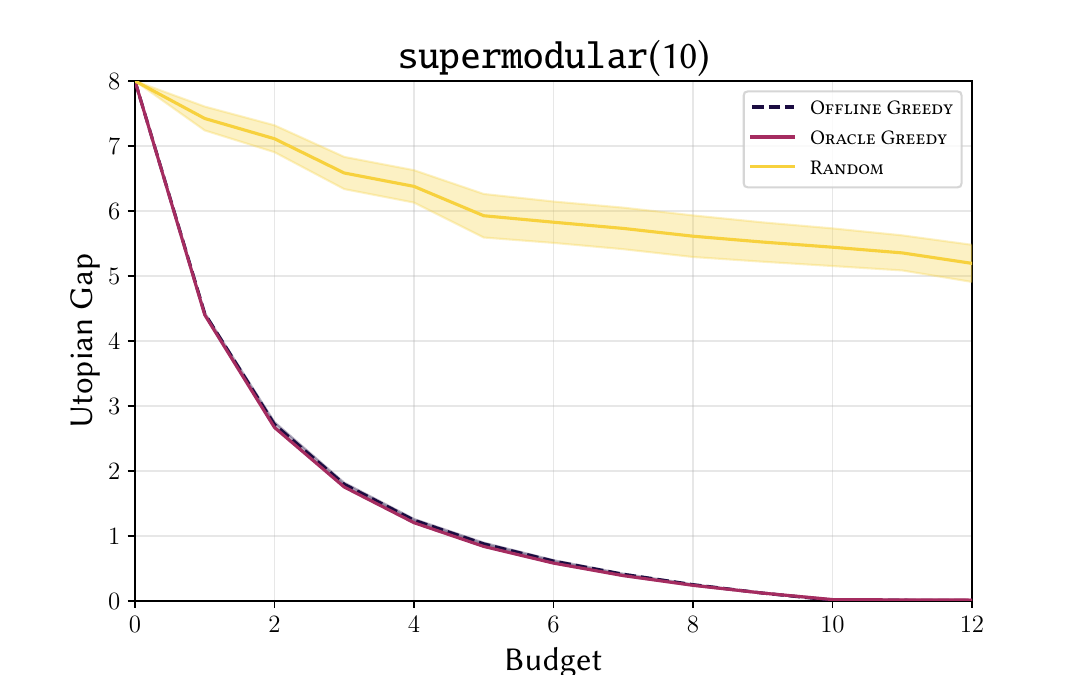}
	\caption{
		The utopian gap as a function of number of revealed coalitions (i.e. steps of the Principals problem) for different sizes.
		Note that the {\sc PPO} algorithm was trained on the same number of time steps as for the 4 player problem.
		Even though the action space size has increased dramatically (it grows exponentially with $ n $), it still beats the trivial random approach.
		Better performance of the {\sc PPO} approach can still be reached by increasing the learning period.
		Note also, that due to time constraints, when gathering the data for the plots, we used 100 samples for sizes 8 and 9, and 10 samples for size 10, while also dropping the {\sc PPO}, as the training process was too time demanding.
	}
	\label{fig: exploitability extensive}
\end{figure*}
\begin{figure*}[t]
	\vspace{5ex}
	\centering
	\includegraphics[width=\textwidth]{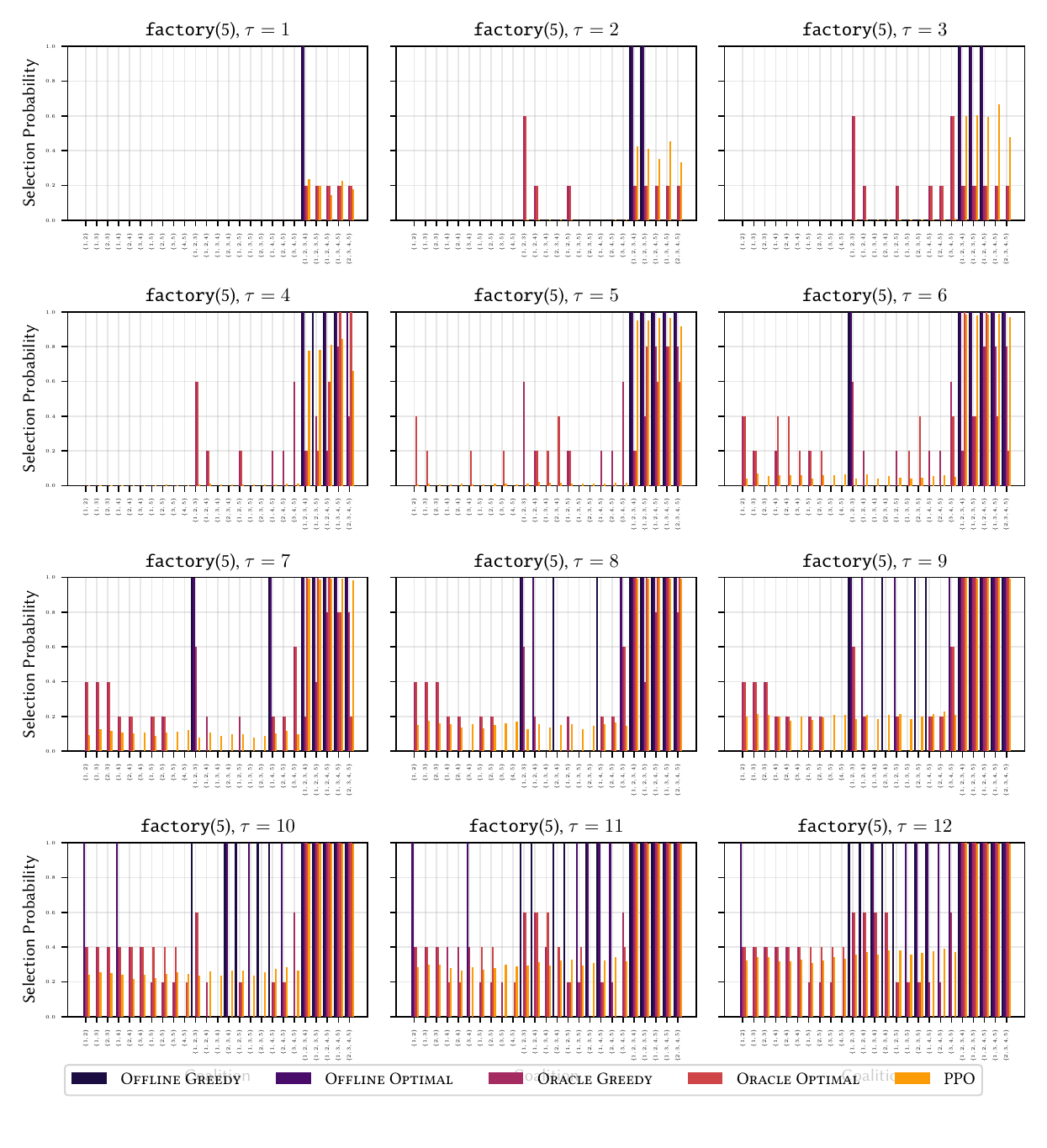}
	\caption{
		Percentage of coalitions selected up to step twelve for {\tt factory}(5) and each algorithm.
		Results show clear preference for larger coalitions, i.e. they contribute more information about the cooperative game on average.
		The oracle algorithms favor smaller coalitions earlier, suggesting the representation of a specific game can efficiently use even smaller coalitions.
			{\sc PPO} initially behaves similarly to the offline algorithms.
		At later steps, it uses the previously obtained values and its selections resemble the oracle methods.
		See Figure~\ref{fig: factory5 bar cum plots} for plot showing individual coalitions.
	}
	\label{fig: app factory5 bar plots}
	\vspace{5ex}
\end{figure*}

\begin{figure*}[t]
	\centering
	\includegraphics[width=\textwidth]{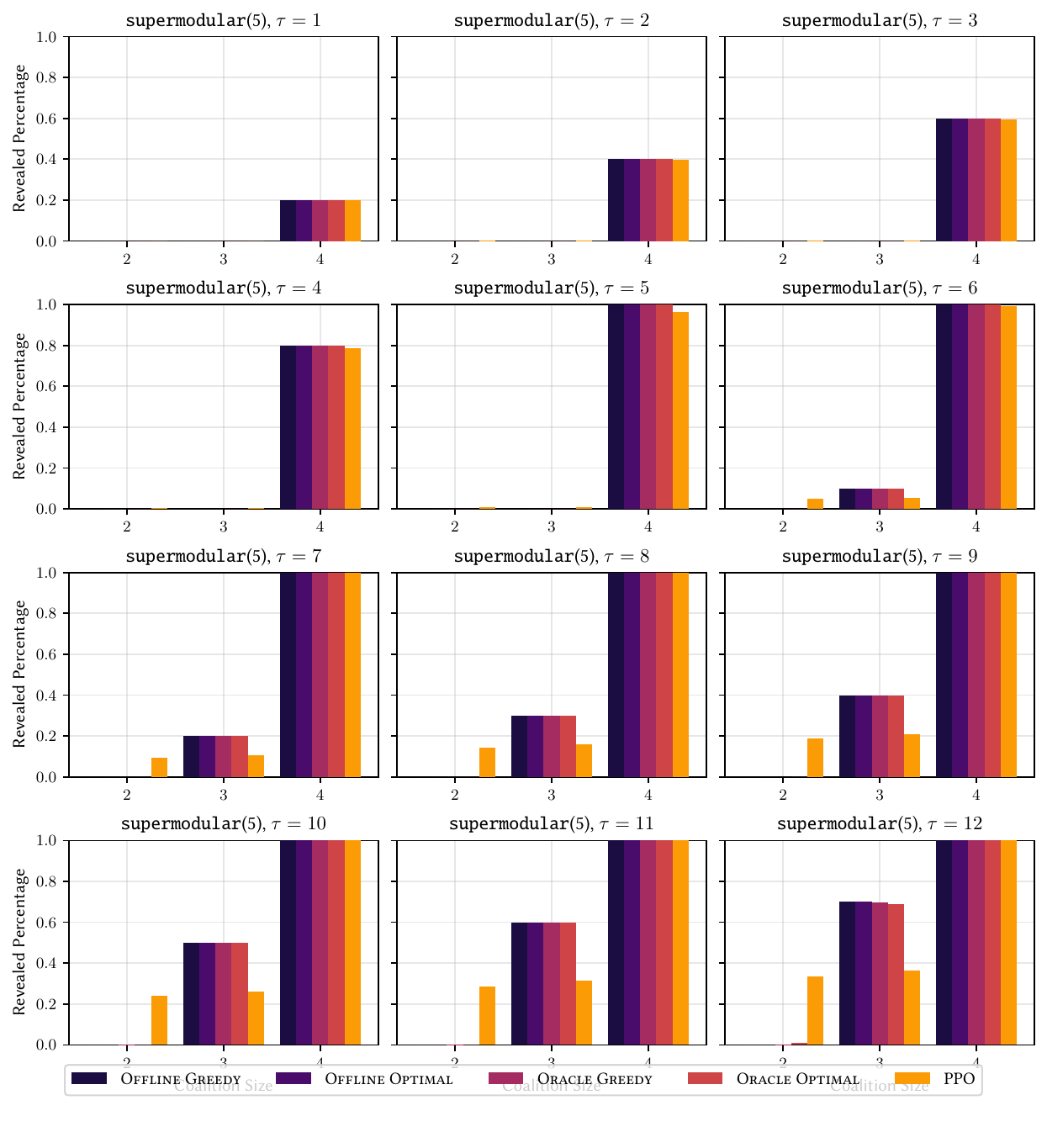}
	\caption{
		Percentage of coalitions selected up to step twelve for {\tt supermodular}(5) and each algorithm.
		Results show clear preference for larger coalitions, i.e. they contribute more information about the cooperative game on average.
		The oracle algorithms favor smaller coalitions earlier, suggesting the representation of a specific game can efficiently use even smaller coalitions.
			{\sc PPO} initially behaves similarly to the offline algorithms.
		At later steps, it uses the previously obtained values and its selections resemble the oracle methods.
		See Figure~\ref{fig: app convex5 bar coalition} for plot showing individual coalitions.
	}
	\label{fig: app convex5 bar coalition sum}
\end{figure*}

\begin{figure*}[t]
	\centering
	\includegraphics[width=\textwidth]{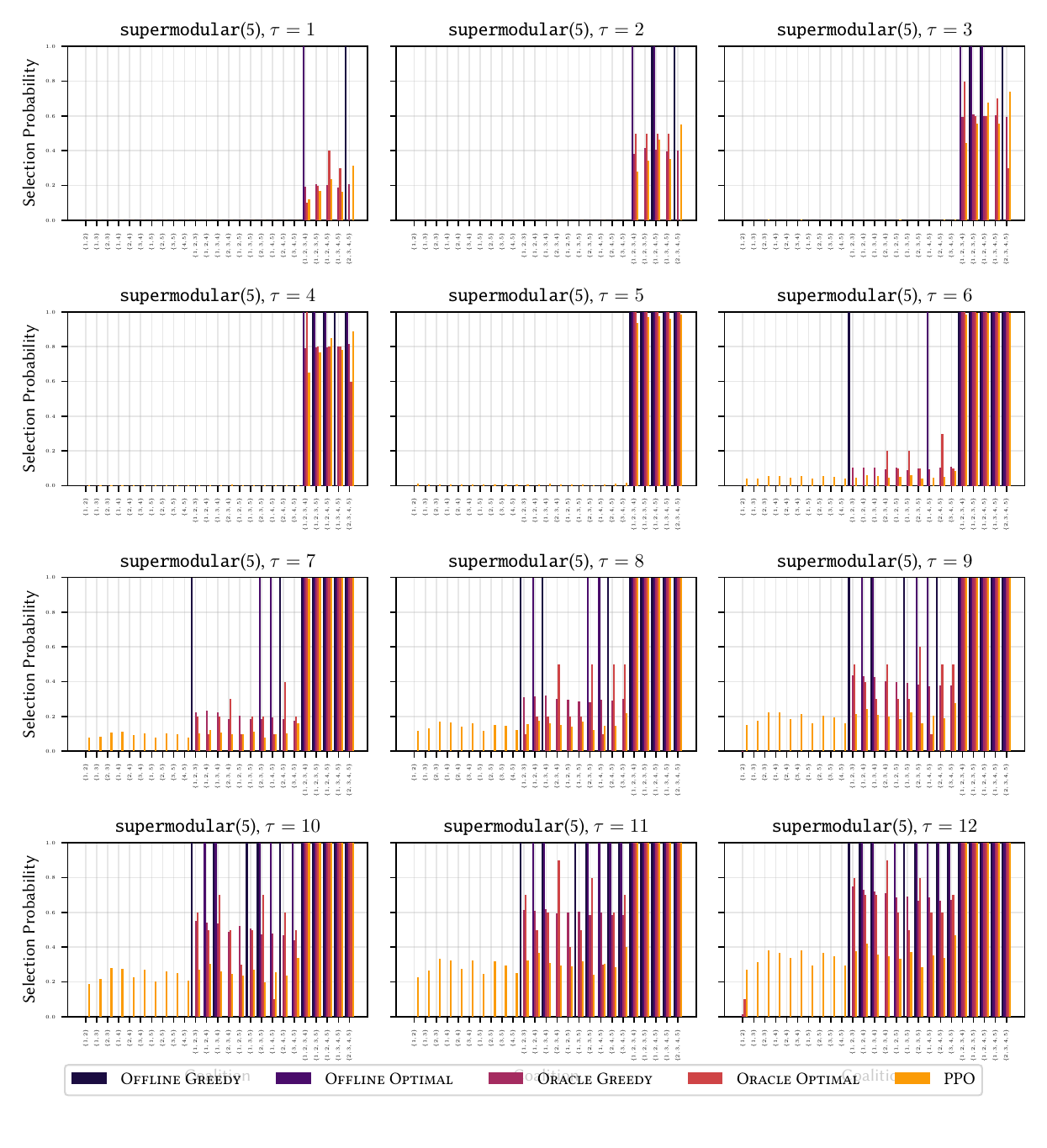}
	\caption{
		Percentage of coalitions selected up to step twelve for {\tt supermodular}(5) and each algorithm.
		Results show clear preference for larger coalitions, i.e. they contribute more information about the cooperative game on average.
		The oracle algorithms favor smaller coalitions earlier, suggesting the representation of a specific game can efficiently use even smaller coalitions.
			{\sc PPO} initially behaves similarly to the offline algorithms.
		At later steps, it uses the previously obtained values and its selections resemble the oracle methods.
		See Figure~\ref{fig: app convex5 bar coalition sum} for plot showing individual coalitions.
	}
	\label{fig: app convex5 bar coalition}
\end{figure*}

In this section, we present studies of classes of superadditive cooperative games which were omitted from the main text due to space constraints.
In Section~\ref{app: factory additional results}, we present results extending Section~\ref{sec: experiments}.
In the remainder of this section, we present focus on other class of superadditive games.
All results were obtained in the same way as those presented in Section~\ref{sec: experiments}.

In this section, we extend the results presented in Section~\ref{sec: experiments} by showing how the utopian gap depends on the number of revealed coalitions for four players.
Our results are presented in Figure~\ref{fig: app exploitability over time}.
The same general trends seen in the main text are present here.
The oracle algorithms outperform their offline counterparts, especially at for lower values of $t$.
At $t=4$, the optimal strategy is to uncover values of all coalitions of size three, which is impossible for {\sc Oracle Greedy}, see Appendix~\ref{app: local is not global example}.
In general, {\sc PPO} initially performs similarly to the offline algorithms since it has no information about the underlying game.
At later steps, it can leverage the known values to outperform the offline algorithms, nearly matching the performance of the {\sc Online Optimal}.

For {\tt supermodular}($4$), revealing values of the largest coalition first rapidly decreases the utopian gap.
In this case, revealing all coalitions of size three decreases the utopian gap by $\approx 99\%$.

Next, we show in Figure~\ref{fig: app factory5 bar plots} a variation of Figure~\ref{fig: factory5 bar cum plots}, where we do not sum over coalitions of the same size.
The results show that all algorithms prefer to select the largest coalitions first.
In later steps, a weak preference towards coalitions of size three emerges.
Since there are more coalitions of size three than two, Figure~\ref{fig: app convex5 bar coalition sum} suggested there is a preference for those coalitions.

Finally, we present the study of percentage of selected coalitions for the {\tt supermodular}(5) class of games.
Our results are presented in Figure~\ref{fig: app convex5 bar coalition sum}~and~\ref{fig: app convex5 bar coalition}
All algorithms show preference towards larger coalitions initially.
At later steps, {\sc PPO} deviates from the other strategies and selects coalitions pseudo-uniformly, while the rest favors coalitions of size three.
We speculate this is because the utopian gap is very small at this stage, making the actions seem near equal in terms of the reward, i.e. the utopian gap.

\subsection{Graph Games}
\label{app: graph games}

\begin{figure*}[t]
	\centering
	\includegraphics[width=0.3\textwidth]{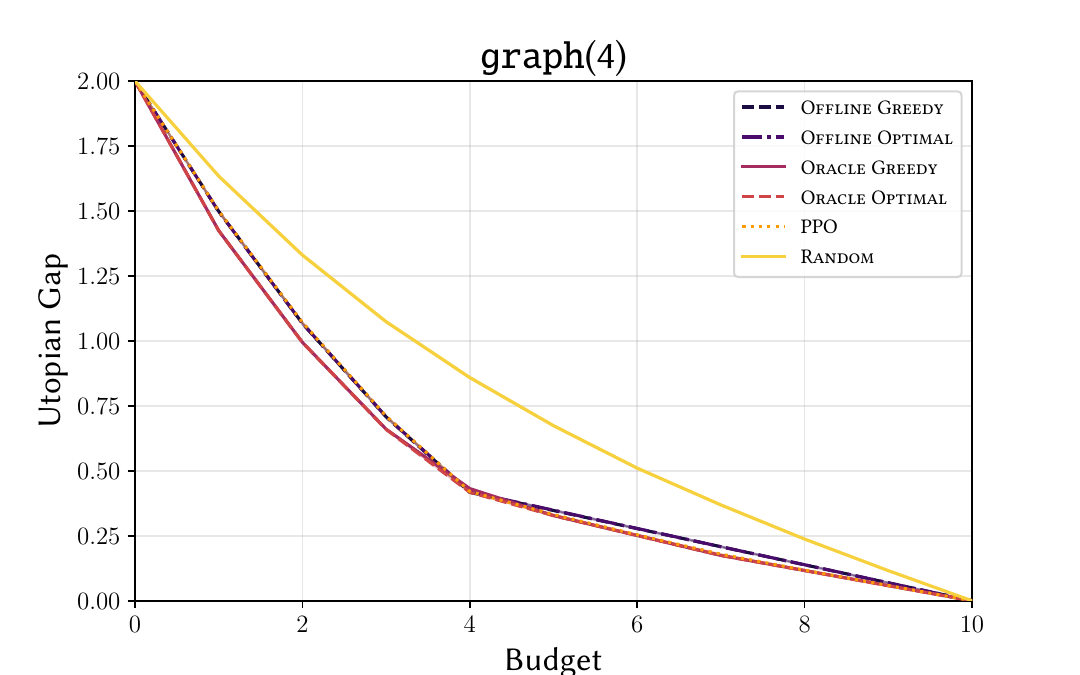}
	\includegraphics[width=0.3\textwidth]{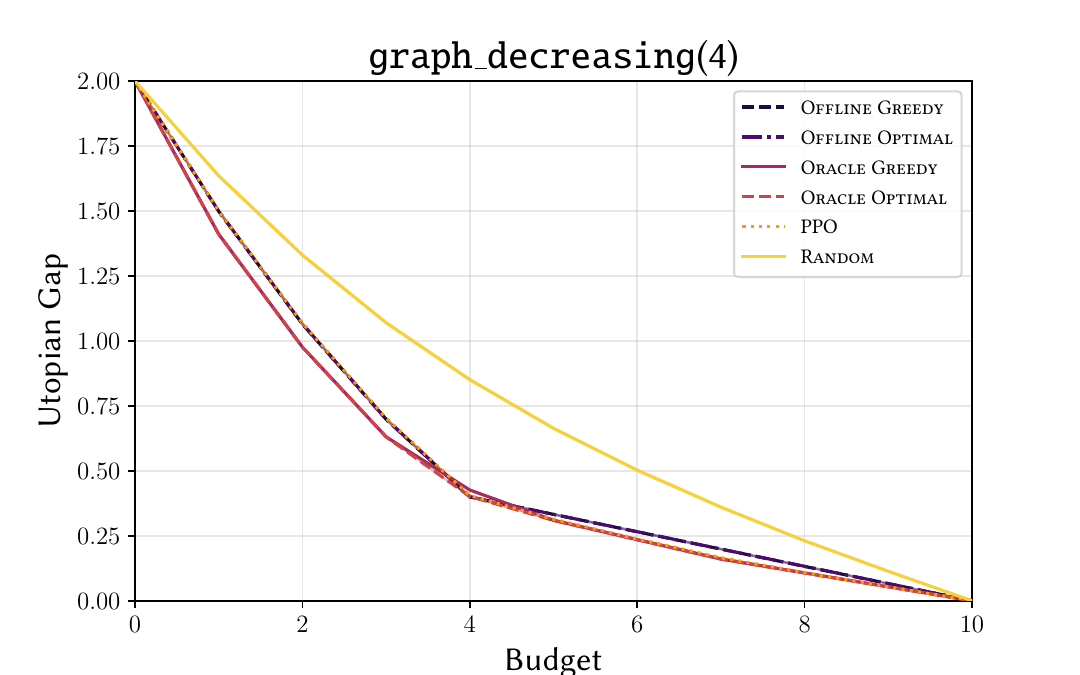}
	\includegraphics[width=0.3\textwidth]{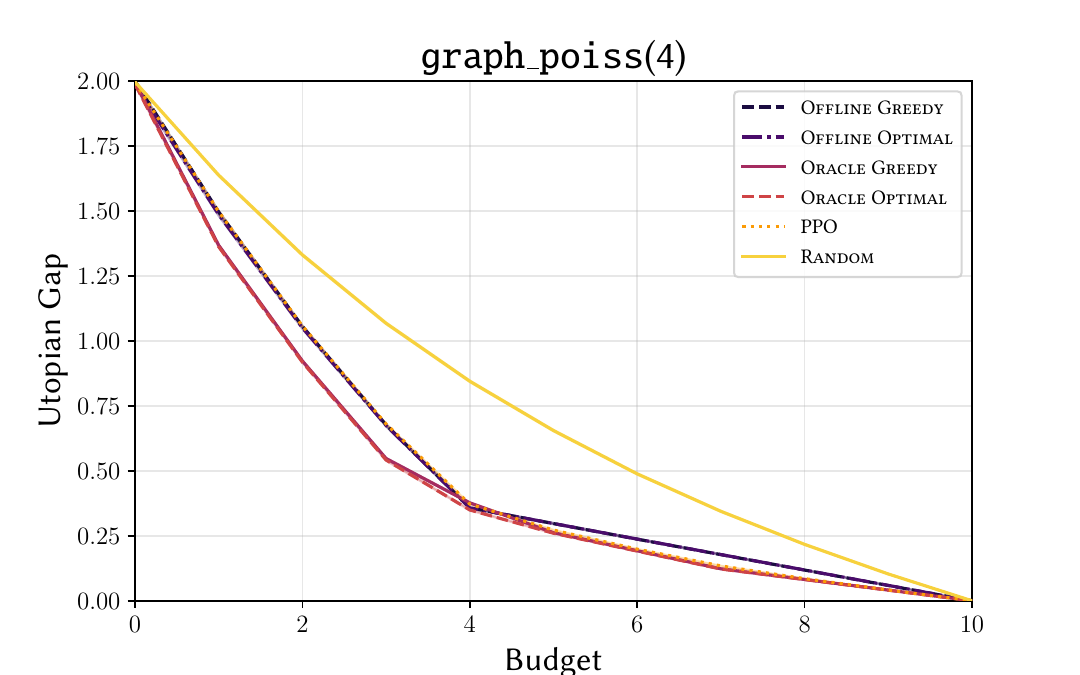}
	\caption{The cumulative utopial gap as a function of the number of revealed coalitions on graph cooperative games. We use {\tt graph}(4) (left), {\tt graph\_decreasing}(4) (middle) and {\tt graph\_poiss}(4) (right) distributions to generate the weights of the graph. See Appendix~\ref{app: graph games} for more details.}
	\label{fig: app graph games}
\end{figure*}

\begin{figure*}[t!]
	\centering
	\includegraphics[width=0.3\textwidth]{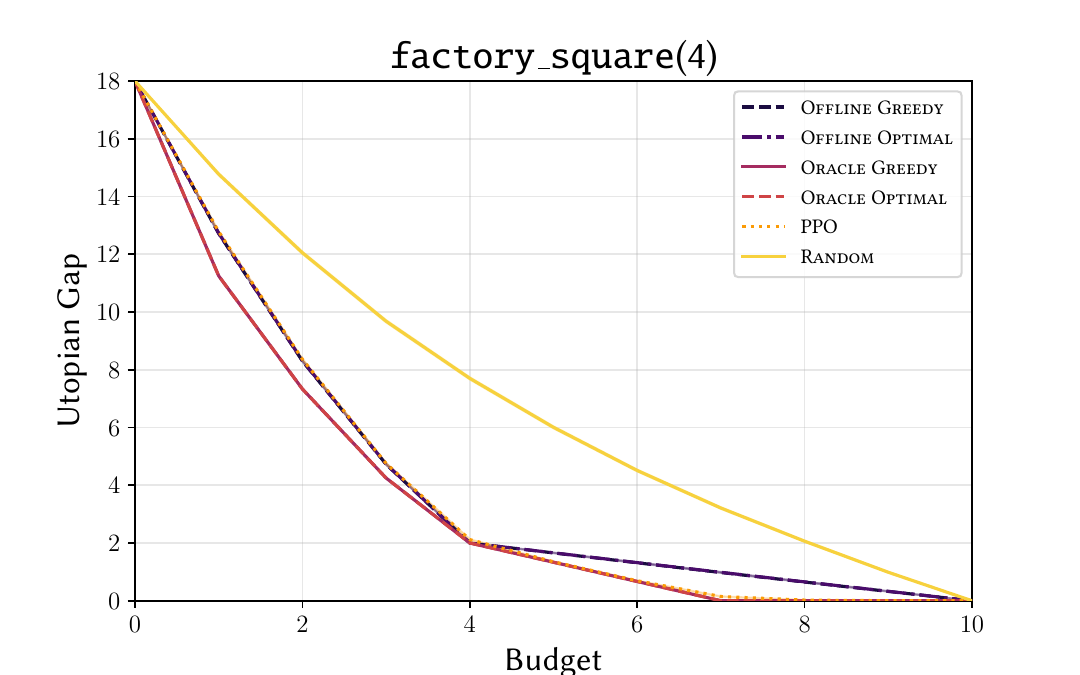}
	\includegraphics[width=0.3\textwidth]{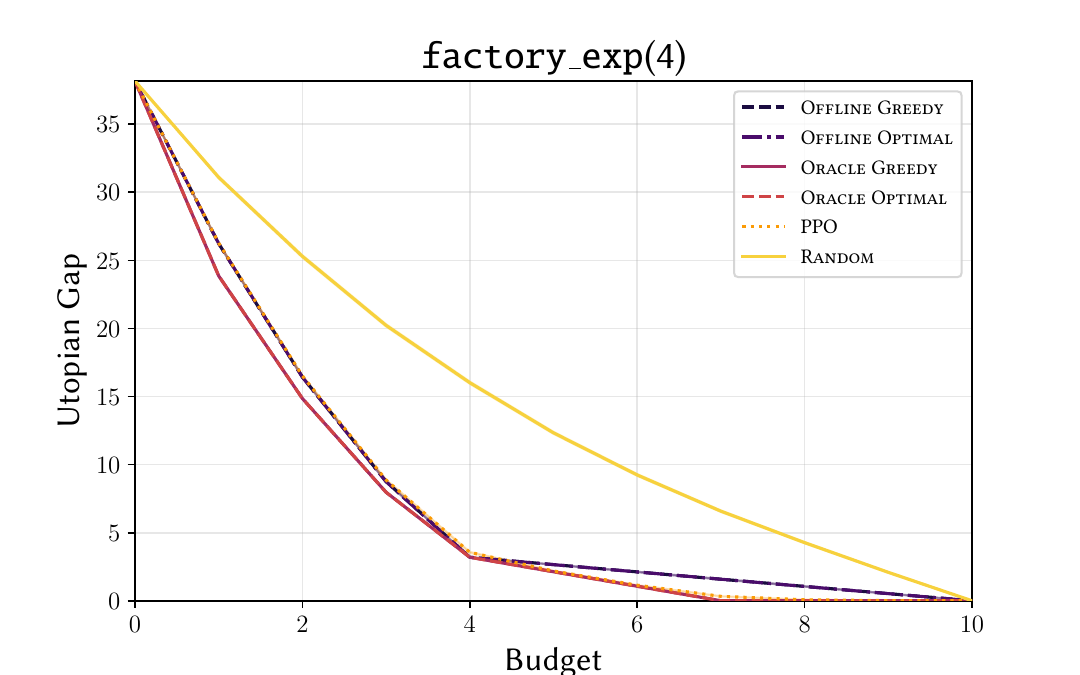}
	\includegraphics[width=0.3\textwidth]{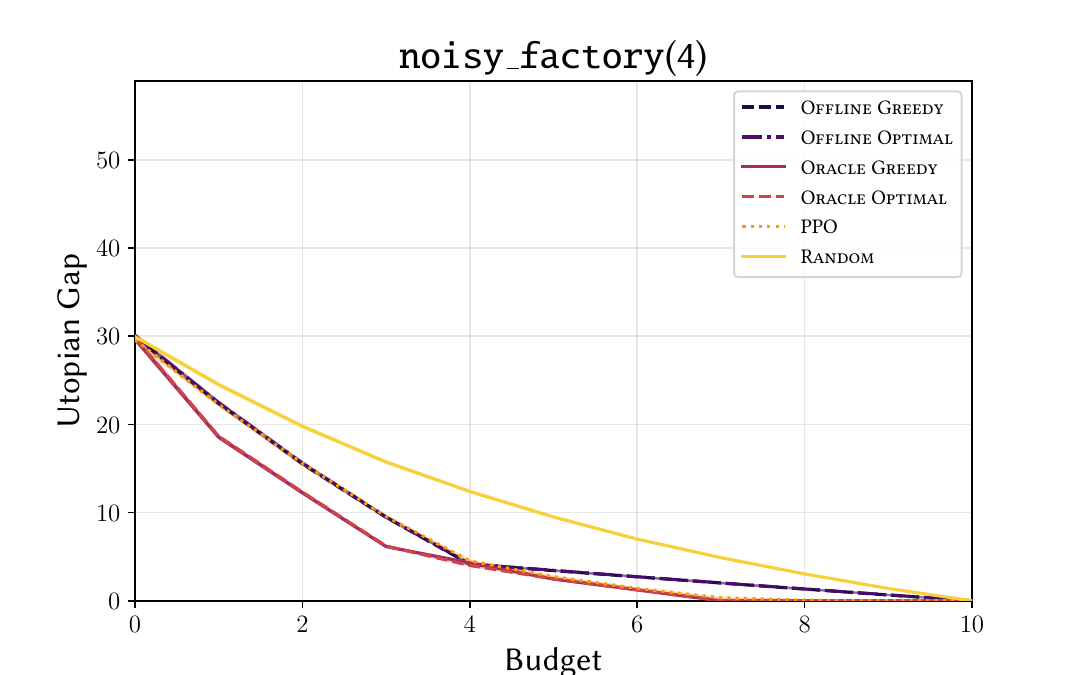}
	\caption{The cumulative utopial gap as a function of the number of revealed coalitions on graph cooperative games. We use {\tt factory\_square}(4) (left), {\tt factory\_exp}(4) (middle) and {\tt noisy\_factory}(4) (right) distributions to generate the weights of the graph. See Appendix~\ref{app: modifications of factory} for more details.}
	\label{fig: app modifications of factory}
\end{figure*}
Consider a complete weighted graph $G=(N, E)$ on the set of players $N$. Then we define the value of characteristic function for a coalition $S$ as the sum of weights of the subgraph induced by $S$

\begin{equation*}
	v(S) = \sum_{i,j \in S}w_{i,j},
\end{equation*}
where $w_{i,j}\ge 0$ is the weight of $e_{i,j}$.
We construct several distributions $\valueDistribution$ based on how $w$'s are generated; in each case, i.i.d.

We investigate three distributions in particular. The {\tt graph}($n$) generates weights from the uniform distribution $w\sim\mathcal{U}(0,1)$.
The {\tt graph\_decreasing}($n$) generates weights from a distribution, whose probability density function is $ 2-2x $ on the interval $ \left[ 0,1 \right] $.
Finally, the {\tt graph\_poiss}($n$) generates weights from the Poisson distribution with $\lambda=1$.

We focus on $n=4$ for these experiments.
The evolution of the utopian gap as a function of the number of revealed coalitions for each algorithm is presented in Figure~\ref{fig: app graph games}.
While all algorithms outperform the {\sc Random} strategy, they all give similar results.
The oracle methods are slightly more efficient at minimizing the utopian gap compared to the offline algorithms.
	{\sc PPO} is initially close to the offline methods, and later outperforms them, showing it can utilize the limited knowledge of the underlying game.

\subsection{Modifications of Factory}
\label{app: modifications of factory}
In this section, we discuss three modification of the {\tt factory} class defined in the main text.
In the {\tt factory\_square}($n$) all values of coalitions are squared
\begin{equation*}
	v(S) =
	\begin{cases}
		(|S|-1)^2 & \text{if } o \in S, \\
		0         & \text{otherwise}.
	\end{cases}
\end{equation*}
Similarly, the {\tt factory\_exp}($n$) makes all values exponentially larger
\begin{equation*}
	v(S) =
	\begin{cases}
		\exp(|S|-1) & \text{if } o \in S, \\
		1           & \text{otherwise}.
	\end{cases}
\end{equation*}
Finally, in the the {\tt noisy\_factory}($n$), each worked $w_i$ has a random productivity $p_i\sim\mathcal{U}(0,1)$.
The value of a coalition $S\in2^N$ is
\begin{equation*}
	v(S) =
	\begin{cases}
		\sum_{i\in S}p_i & \text{if } o \in S, \\
		0                & \text{otherwise},
	\end{cases}
\end{equation*}
where the owner's productivity is $p_o=0$.

We can draw similar conclusions as in the main text.

\section{Oracle algorithms}\label{app:oracle-algs}
Following are the pseudocodes of Oracle algorithms discussed in Section 4.
\begin{algorithm}[h!]
	\caption{\sc Oracle Optimal}
	\label{algo: online optimal}
	\SetAlgoLined
	\DontPrintSemicolon
	\KwIn{characteristic function $v\in \S$, number of steps $t$}
	\vspace{1ex}
	$ \overline \K \gets 2^N \setminus \K_0 $\;
	$ \left\{ S_i \right\}_{i=1}^t \gets \argmin_{\mathcal S \subseteq \overline \K: \absolute{\mathcal S} = t} \utopianGap_{(N,v)}(\mathcal{S})$\;
	\Return{$\{S_i\}_{i=1}^{t}$}\;
\end{algorithm}

\begin{algorithm}[h!]
	\caption{\sc Oracle Greedy}
	\label{algo: online greedy}
	\SetAlgoLined
	\DontPrintSemicolon
	\KwIn{characteristic function $v\in \S$, number of steps $t$}
	\vspace{1ex}
	\SetKwFunction{FNC}{Oracle Greedy}
	\SetKwProg{Fn}{Function}{:}{}
	$\{S_i\}_{i=1}^{t - 1} \gets $ {\FNC{$v, t-1$}}\;
	$ \overline \K \gets 2^N \setminus (\K_0 \cup \left\{ S_i \right\}_{i=1}^{t - 1} $)\;
	$S_{t} \gets \argmin_{S\in \overline \K} \utopianGap_{(N,v)}(\{S_i\}_{i=1}^{t - 1}\cup \left\{ S \right\} ,v)$\;
		\Return{$\{S_i\}_{i=1}^{t}$}\;
\end{algorithm}

\section{Value Function Normalization}\label{app: value normalization}

In our application, it is convenient to transform a cooperative game $(N, v)$ by an affine mapping such that after the transformation, the values of the singletons are equal to 0 and the value of the grand coalition is equal to 1.
Affine transformations correspond to \emph{strategic equivalence} in cooperative game theory. We say two games $(N,v)$ and $(N,w)$ are \emph{strategically equivalent} if there is $\alpha > 0$ and $\beta \in \mathbb{R}^n$ such that
\begin{equation}\label{eq:SE-transform}
	w(S) = \alpha \cdot v(S) + \sum_{i \in S}\beta_i.
\end{equation}
By considering
\begin{equation*}
	\beta_i = \frac{v(\{i\})}{v(N) - \sum_{i \in N}v(\{i\})}
	\hspace{3ex}\text{and}\hspace{3ex}
	\alpha = \frac{1}{v(N)},
\end{equation*}
we achieve the desired transformation. Lemma~\ref{lem:gap-is-SE} describes the effect on the utopian gap, when such a transformation is applied.

\begin{lemma}\label{lem:gap-is-SE}
	The cumulative utopian gap $\utopianGap$ is in the characteristic function \textit{invariant under the strategic equivalence}, meaning
	\begin{equation}
		\utopianGap_{(\alpha v + \beta)}(\K) = \alpha \cdot \utopianGap_{(N,v)}(\K)
	\end{equation}
	where $\alpha > 0$,  $(\alpha v + \beta)(S) \coloneqq \alpha v(S) + \sum_{i \in N}\beta_i$ and $\beta_i \in \mathbb{R}$ for $i \in N$.
\end{lemma}

\begin{proof}
	It is a standard result~\cite{Peleg2007} that
	\begin{enumerate}
		\item $(N,v) \in \S \implies (N,\alpha v + \beta) \in \S$,
		\item $\phi_i(\alpha v + \beta) = \alpha\phi_i(v) + \beta_i$  for every $i \in N$.
	\end{enumerate}
	From these two results, we have
	\begin{align*}
		\utopianGap_{(N,\alpha v + \beta)}(\K)
		 & =
		\sum_{i\in N}\left(\max\limits_{w\in \S(\K,\alpha v + \beta)} \phi_i(w) \right) - (\alpha v +  \beta)(N) \\
		 & =
		\sum_{i\in N}\left[ \phi_i((\alpha v + \beta)_i) - \phi_i(\alpha v + \beta) \right]                      \\
		 & =
		\sum_{i\in N}\left[\alpha \phi(v_i) + \beta_i - \left(\alpha\phi_i(v) + \beta_i\right) \right]           \\
		 & =
		\sum_{i\in N}\alpha\left[\phi_i(v_i) - \phi_i(v)\right]                                                  \\
		 & =
		\alpha\cdot \utopianGap_{(N,v)}(\K)	.
	\end{align*}
\end{proof}

\section{Explicit form of the gap}\label{app:mere-technicality}

Denote $\hat{\K}$, such that $\K_0 \subseteq \hat{\K} \subseteq 2^N$ and $\K = \hat{\K} \setminus \K_0$ and recall the definition of the gap, i.e.

\[\utopianGap_{(N,v)}(\K) = \sum_{i \in N}\left(\max\limits_{w \in \S(\K \cup \K_0,v)}\phi_i(w)\right) - v(N).\]

In Proposition~\ref{prop:gap-is-utopian}, we showed it holds $\max\limits_{w \in \S(\K \cup \K_0,v)}\phi_i(w) = \phi_i(v_i)$, thus

\begin{equation}\label{eq:mere-technicality}
	\utopianGap_{(N,v)}(\K) = \sum_{i \in N}\sum_{S \subseteq N\setminus i}\beta_{|S|}\left(\overline{v}_{\hat{\K}}(S \cup i) - \underline{v}_{\hat{\K}}(S)\right)-v(N),
\end{equation}
where $\beta_{|S|} = \frac{|S|!(|N|-|S|-1)!}{|N|!}$. We rewrite this double sum as a sum of $S \subseteq N$. First, $\overline{v}_{\hat{\K}}(S)$ appears in~\eqref{eq:mere-technicality} for every $i \in S$ with coefficient $\beta_{|S|-1}$.
Second, $\underline{v}_{\hat{\K}}(S)$ appears in~\eqref{eq:mere-technicality} for every $i \in N \setminus S$ with coefficient $\beta_{|S|}$. We get for any $i \in S$,
\begin{equation*}\label{eq:mere-technicality-2}
	\utopianGap_{(N,v)}(\K) = \sum_{S \subseteq N}\left(\gamma_{|S|}\overline{v}_{\hat{\K}}(S) - \delta_{|S|}\underline{v}_{\hat{\K}}(S)\right)-v(N),
\end{equation*}
where
\begin{itemize}
	\item $\gamma_{|S|} = |S|\frac{(|S|-1)!(|N|-|S|)!}{|N|!} = \frac{|S|!(|N|-|S|)!}{|N|!}$,
	\item $\delta_{|S|} = (|N|-|S|) \frac{|S|!(|N|-|S|-1)!}{|N|!} = \frac{|S|!(|N|-|S|)!}{|N|!}$.
\end{itemize}
Notice there are two edge-cases, where the second equality does not hold:
\begin{enumerate}
	\item If $ S = \emptyset $, then $ \gamma_{0} = 0 $, which is lost when $ S $ joins the factorial expression, yielding $ 0! = 1 $. However, $ \gamma_0 $ is always multiplied by $ \left( \overline v_{\hat \K}(S) - \underline v_{\hat \K}(S) \right) = \left( 0-0 \right) = 0 $, so the value of this coefficient does not change the resulting value.
	\item If $ S = N $, then $ \gamma_{n} = 0 $, however, $\frac{|S|!(|N|-|S|)!}{|N|!}=1$.
		In fact, $ \overline v_{\hat \K}(S) $ appears exactly once in the sum, but $ \underline v_{\hat \K}(S) $ does not appear at all.
		Our original calculations, however, didn't take into account the final $ -v(N) $ term present at the end of the definition of the gap.
		With the subtraction of $ v(N) = \underline v_{\hat \K}(S) $, it is easy to see that the value of $ \gamma_n = 1 $ is correct.
\end{enumerate}

This concludes the proof of Proposition~\ref{prop:gap-is-utopian}, as
\[
	\utopianGap_{(N,v)}(\K) = \sum_{S \subseteq N}\frac{|S|!(|N|-|S|)!}{|N|!}\left(\overline{v}_{\hat{\K}}(S) - \underline{v}_{\hat{\K}}(S)\right).
\]

\section{Supermodularity of gap}\label{app:counter-examples}
Supermodularity of the utopian gap is a desirable property, as it is known that greedy algorithms perform well on supermodular function~\cite[Proposition 3.4]{Nemhauser1978}. We show that as long as $|N| \le 4$, the gap is supermodular, however, when $|N|\ge 5$, the supermodularity does not have to hold for different subsets of superadditive games. For $|N| = 5$, we give examples from different important subsets of superadditive games that do not satisfy supermodularity of the gap. Among these are totally monotonic (thus also supermodular), symmetric, and graph games. For $|N|\geq 6$, we derive a criterion, which, when satisfied, implies non-supermodularity of the gap. As a corollary, we show all three of the above mentioned classes contain games with non-supermodular gap.

Throughout this section, we implicitly assume that the value function $v: 2^N \to \mathbb{R}$ is normalized as described in Appendix~\ref{app: value normalization}.
This is without the loss of generality, as this preserves the value of the gap, see Lemma~\ref{lem:gap-is-SE}.
As a consequence, the lower function $\underline{v}(S)$ is zero for $\K = \K_0$ for any $S \in 2^N \setminus \K_0$.

\subsection{Supermodularity for $|N| \leq 4$}
\label{app: supermodularity for n<5}
Before we prove Proposition~\ref{prop:gap-supermodularity}, we state several useful lemmata.
\begin{lemma}
	\label{lem:bound-funcs-monotone}
	Let $\K_0 \subseteq \hat{\K}\subseteq 2^N \setminus S$ and $(N,v)\in \S$. Then $\forall T\in 2^N$
	\begin{equation*}
		\underline{v}_{\hat{\K}}(T) \le \underline{v}_{\hat{\K} \cup S}(T)
		\hspace{2ex}
		\text{and}
		\hspace{2ex}
		\overline{v}_{\hat{\K}}(T) \ge \overline{v}_{\hat{\K} \cup S}(T).
	\end{equation*}
\end{lemma}
\begin{proof}
	By definition of the lower game, a larger set of known coalitions can only increase the value of the lower game.
	Consequently, the upper gane cannot increase if the set of known values becomes larger.
\end{proof}

\begin{lemma}
	\label{lem:change-of-bounds}
	Let $\K_0 \subseteq\hat{\K} \subseteq 2^N \setminus S$  and $T\in 2^N$.
	Then
	\begin{equation*}
		T\subsetneq S \implies \overline{v}_{\hat{\K}}(T) = \overline{v}_{\hat{\K}\cup S}(T),
	\end{equation*}
	and

	\begin{equation*}
		S\subsetneq T \implies \underline{v}_{\hat{\K}}(T) = \underline{v}_{\hat{\K}\cup S}(T).
	\end{equation*}

\end{lemma}
\begin{proof}
	The value of $T$ for $\overline{v}_{\hat{\K}}$, resp. $\overline{v}_{\hat{\K} \cup S}$ is given by minimal partition of $T$ in $\K$, resp. $T$ in $\hat{\K} \cup S$. As $T \subsetneq S$, revealing $f(S)$ does not appear in any partition of $T$, thus $\overline{v}_{\hat{\K}}(T) = \overline{v}_{\hat{\K}\cup S}(T)$.

	The value of $T$ for $\underline{v}_{\hat{\K}}$, resp. $\underline{v}_{\hat{\K} \cup S}$ is given by considering maximum over $X \in \hat{\K}$, $T \subseteq X$. As $S \subsetneq T$, it does not appear as a possible $X$. The only possibility for the change of the upper value would be that by reavealing $S$, we get $\overline{v}_{\hat{\K}\cup S}(X \setminus T) \neq \overline{v}_{\hat{\K}}(X \setminus T)$ for some $X$. This means $S \subseteq X \setminus T$, but this is not possible as $S \subsetneq T$.
\end{proof}

\begin{lemma}
	\label{lem:unaffected-bounds}
	Let $\K_0 \subseteq\hat{\K} \subseteq 2^N\setminus S$ and $T\in 2^N$, such that $\emptyset \neq S\cap T$, $S\subsetneq T$ and $T \not\subseteq S$. Then
	\begin{equation*}
		\overline{v}_{\hat{\K}}(T) = \overline{v}_{\hat{\K}\cup S}(T),
		\hspace{3ex}
		\text{and}
		\hspace{3ex}
		\underline{v}_{\hat{\K}}(T) = \underline{v}_{\hat{\K}\cup S}(T).
	\end{equation*}
\end{lemma}
\begin{proof}
	Follows immediately from the definition of the upper/lower games.
\end{proof}
\begin{cmr}{Proposition~\ref{prop:gap-supermodularity}}
	For $|N|\le 4$, the gap of an incomplete $\S$-extendable game is supermodular.
\end{cmr}
\begin{proof}
	Denote $\hat{\K}$, such that $\K_0 \subseteq \hat{\K} \subseteq 2^N$ and $\K = \hat{\K} \setminus \K_0$. As $\utopianGap_{(N,v)}(\K) = \sum_{T \subseteq N}\alpha_T\Delta_{\hat{\K}}(T)-v(N)$ where $\alpha_T = \frac{|T|!(|N|-|T|)!}{|N|!}$, the condition of supermodularity is
	\[
		\sum_{T \subseteq N}\hspace{-0.1cm}\alpha_T\hspace{-0.05cm}\left(\Delta_{\hat{\K}\cup S\cup Z}(T) - \Delta_{\hat{\K}\cup S}(T)\right)\ge\hspace{-0.2cm}
		\sum_{T \subseteq N}\hspace{-0.1cm}\alpha_T\hspace{-0.05cm}\left(\Delta_{\hat{\K}\cup Z}(T) - \Delta_{\hat{\K}}(T)\right).
	\]
	To prove this proposition, we show an even a stronger condition holds, i.e., $\forall T \in 2^N$
	\[
		\alpha_T\left(\Delta_{\hat{\K}\cup S\cup Z}(T) - \Delta_{\hat{\K}\cup S}(T)\right)\ge
		\alpha_T\left(\Delta_{\hat{\K}\cup Z}(T) - \Delta_{\hat{\K}}(T)\right),
	\]
	or since $\alpha_T \geq 0$ for every $T$, this is equivalent to
	\begin{multline}
		\label{eq:SM-stronger-constraint}
		\overline{v}_{\hat{\K}\cup S\cup Z}(T) - \underline{v}_{\hat{\K}\cup S\cup Z}(T)
		- \overline{v}_{\hat{\K}\cup S}(T) +\underline{v}_{\hat{\K}\cup S}(T)
		\ge\\
		\overline{v}_{\hat{\K}\cup Z}(T) - \underline{v}_{\hat{\K}\cup Z}(T) - \overline{v}_{\hat{\K}}(T) + \underline{v}_{\hat{\K}}(T).
	\end{multline}

	For $T \in \hat{\K} \cup S$, condition~\eqref{eq:SM-stronger-constraint} holds for any superadditive $v\colon 2^N \to \mathbb{R}$ with $N$ of arbitrary size. To see this, notice $\overline{v}_{\hat{\K} \cup S \cup Z}(T) = \underline{v}_{\hat{\K} \cup S \cup Z}(T)$ and $\overline{v}_{\hat{\K} \cup S}(T) = \underline{v}_{\hat{\K} \cup S}(T)$, thus~\eqref{eq:SM-stronger-constraint} reduces to
	\[
		\overline{v}_{\hat{\K}\cup S}(T) - \underline{v}_{\hat{\K}\cup S}(T) \le \overline{v}_{\hat{\K}}(T) - \underline{v}_{\hat{\K}}(T).
	\]
	This inequality follows from Lemma~\ref{lem:bound-funcs-monotone}.

	First, consider case $|N|=3$. Since, only coalitions of size 2 are possibly unknown, their lower and upper bounds do not change when another subset of size 2 is revealed. Thus, for $\K_0 \subseteq \hat{\K} \subseteq 2^N \setminus \{S,Z\}$ and $T \notin \hat{\K}\cup S$, it holds
	\[\overline{v}_{\hat{\K} \cup S \cup Z}(T) - \underline{v}_{\hat{\K}\cup S\cup Z}(T) = \overline{v}_{\hat{\K} \cup S}(T) - \underline{v}_{\hat{\K} \cup S}(T)\]
	and similarly
	\[\overline{v}_{\hat{\K} \cup Z}(T) - \underline{v}_{\hat{\K}\cup Z}(T) = \overline{v}_{\hat{\K}}(T) - \underline{v}_{\hat{\K}}(T),\]
	which means~\eqref{eq:SM-stronger-constraint} holds.

	Now, consider caser $|N|=4$. To prove~\eqref{eq:SM-stronger-constraint} for $\K_0 \subseteq \hat{\K} \subseteq 2^N \setminus \{S,Z\}$ and $T \notin \hat{\K} \cup S$, we distinguish several cases based on the relation between $Z$ and $T$. All of the cases follow a similar pattern.
	\begin{enumerate}
		\item $Z \subsetneq T$:

		      Since $\lvert N \rvert = 4$ and $Z,T$ are unknown in $\hat{\K} \cup S$, they must be of form $Z = \{i,j\}$ and $T = \{i,j,k\}$. By Lemma~\ref{lem:change-of-bounds}, $\underline{v}_{\hat{\K}\cup Z}(T) = \underline{v}_\K(T)$ and $\underline{v}_{\hat{\K}\cup S \cup Z}(T) = \underline{v}_{\hat{\K}\cup S}(T)$, thus~\eqref{eq:SM-stronger-constraint} reduces to
		      \[
			      \overline{v}_{\hat{\K}\cup S \cup Z}(T) - \overline{v}_{\hat{\K}\cup S}(T) \ge \overline{v}_{\hat{\K}\cup Z}(T) - \overline{v}_{\K}(T).
		      \]
		      For a contradiction, suppose the converse holds. As, by Lemma~\ref{lem:bound-funcs-monotone}, $\overline{v}_{\hat{\K}\cup S \cup Z}(T) - \overline{v}_{\hat{\K}\cup Z}(T) \ge 0$, this means $\overline{v}_{\hat{\K}\cup Z}(T) > \overline{v}_{\K}(T)$, which leads, by Lemma~\ref{lem:bound-funcs-monotone}, to a contradiction.

		\item $T \subsetneq Z$:
		      Similarly to the first case, we have $T = \{i,j\}$, $Z = \{i,j,k\}$ and by Lemma~\ref{lem:change-of-bounds},~\eqref{eq:SM-stronger-constraint} reduces to
		      \[
			      \underline{v}_{\hat{\K}\cup S}(T) - \underline{v}_{\K\cup S\cup Z}(T) \ge \underline{v}_\K(T)-\underline{v}_{\hat{\K}\cup Z}(T).
		      \]
		      For a contradiction, suppose the converse holds. It means $\underline{v}_\K(T) > \underline{v}_{\hat{\K}\cup Z}(T)$, which by Lemma~\ref{lem:bound-funcs-monotone} leads to a contradiction.
		\item $Z = T$:
		      It holds $\overline{v}_{\hat{\K}\cup S \cup Z}(T) = \underline{v}_{\hat{\K}\cup S\cup Z}(T)$ and $\overline{v}_{\hat{\K}\cup Z}(T) = \underline{v}_{\hat{\K}\cup Z}(T)$, thus~\eqref{eq:SM-stronger-constraint} reduces to
		      \[
			      \overline{v}_{\hat{\K}}(T) - \underline{v}_\K(T) \ge \overline{v}_{\hat{\K}\cup S}(T) - \underline{v}_{\hat{\K}\cup S}(T)
		      \]
		      which holds by Lemma~\ref{lem:bound-funcs-monotone}.
		      In this case, $Z = \{i,j\}$, $T = \{k,\ell\}$ and~\eqref{eq:SM-stronger-constraint} reduces to
		      \[
			      \underline{v}_{\hat{\K}\cup S}(T) - \underline{v}_{\hat{\K}\cup S \cup Z}(T)\ge \underline{v}_{\K}(T) - \underline{v}_{\hat{\K}\cup Z}(T)
		      \]
		      as $\overline{v}_{\hat{\K}\cup S\cup Z}(T) = \overline{v}_{\hat{\K}\cup S}(T)$ and $\overline{v}_{\hat{\K}\cup Z}(T) = \overline{v}_{\hat{\K}}(T)$. For a contradiction, if $<$ holds, similarly to previous cases, $\underline{v}_{\hat{\K}}(T) > \underline{v}_{\hat{\K}\cup S}(T)$, a contradiction with Lemma~\ref{lem:bound-funcs-monotone}.

		\item $Z \cap T \neq \emptyset$ and $Z \not\subseteq T$ and $T \not\subseteq Z$: By Lemma~\ref{lem:unaffected-bounds},
		      condition~\eqref{eq:SM-stronger-constraint} reduces to $0 \ge 0$.
	\end{enumerate}

\end{proof}

\subsection{Non-supermodularity for $|N| = 5$}
All of the examples of games with non-supermodular gap are based on the same principle. We fix $\K = \K_0 \cup \{1,2,3\}$, $S = \{1,2\}$ and $Z = \{3,4\}$. Then if we denote
\[
	\utopianGap_\Delta = \utopianGap_{(N,v)}(\K \cup S \cup Z) - \utopianGap_{(N,v)}(\K \cup S) - \utopianGap_{(N,v)}(\K \cup Z) + \utopianGap_{(N,v)}(\K),
\]
a necessary condition of supermodularity of $\utopianGap_{(N,v)}$ is thus $\utopianGap_\Delta \ge 0$. We note that a necessary, however not sufficient condition for violation of this condition is $v(\{1,2\}) + v(\{3,4\}) > v(\{1,2,3\})$.

\begin{example}{(Totally monotonic game)}

	Totally monotonic games are defined as non-negative combinations of so called \emph{unanimity games}. Unanimity game $(N,u_S)$ for $S \subseteq N$ is defined as
	\[
		u_S(T) = \begin{cases}
			1 & S \subseteq T,    \\
			0 & \text{otherwise.} \\
		\end{cases}
	\]
	Any totally monotonic game can be then expressed as $\sum_{S \subseteq N}\alpha_Su_S$ where $\alpha_S \ge 0$ for every $S$. They form a subset of convex games, which under the name \emph{belief measures} forms the foundation on the \emph{evidence theory}~\cite{Grabisch2016, shafer1976})

	An example of a totally monotonic game with non-supermodular gap is $u_{\{3,4\}} + u_{\{1,2,3\}}$, for which $\utopianGap_\Delta = 0.1 \not\ge 0$.

\end{example}

\begin{example}{(Symmetric game)}
	Symmetric cooperative game $(N,v)$ satisfies for every $S$, $T \subseteq N$, $|S|=|T|$ that $v(S) = v(T)$. Symmetric game of $n$ players can be thus defined by $s_1,\dots,s_n \in \R$ corresponding to values of coalitions of different sizes. These games are often studied for their robust properties with respect to solution concepts~\cite{Peleg2007}).

	It holds that if symmetric game on 5 players satisfy $s_2 = s_3=1$ and $s_5=2$, it follows $\utopianGap_\Delta = -0.1 \not\ge 0$.
\end{example}

\begin{example}{(Graph game)}
	Graph game $(N,v)$ is defined using a complete weighted graph $G = (N,E)$ as
	\[
		v(S) = \sum_{i,j \in S}w_{i,j},
	\]
	where $w_{i,j}\geq 0$ is the weight of $e_{i,j}$. These games have interesting properties from the point of view of algorithmic complexity and have been thoroughly studied in the past decades~\cite{Chalkiadakis2012}.

	For a complete graph $(N,E)$, where $N = \{1,2,3,4,5\}$, $w_{1,2}=1 = w_{3,4}=1$, it holds $\utopianGap_\Delta = -0.1 \not\ge 0$.

\end{example}

\subsection{Criterion for $|N| \geq 6$}\label{app:gap-criterion}
\begin{cmr}{Proposition~\ref{prop:gap-criterion}}
	For $|N|\geq 6$, the gap of an incomplete $\S$-extendable game $(N,w)$ is not supermodular if there exist $i,j,k,l \in N$ such that $v(ij) \le v(jk) \le v(kl)$, where $v(S) = w(S) - \sum_{i\in S}w(\{i\})$, and
	\begin{equation}
		v(kl) < \left[{\binom{n}{2}}{\binom{n}{\lfloor n/2 \rfloor}}^{-1}\left(2^{n-3}-n+2\right)-1\right]v(ij).
	\end{equation}
\end{cmr}
\begin{proof}
	By Lemma~\ref{lem:gap-is-SE} in Appendix~\ref{app: value normalization}, the gap  of $(N,w)$ is supermodular if and only if the gap of $(N,v)$ is supermodular. We focus on the gap $(N,v)$ as this makes our analysis simple and at the same time, the condition easy to check.

	The supermodularity condition \eqref{eq: supermodular constraints} for the gap is
	\begin{equation}
		\label{app: eq: supermodularity constraints}
		\utopianGap_{(N,v)}(\hat{\K}\cup S\cup Z) - \utopianGap_{(N,v)}(\hat{\K}\cup S) \ge \utopianGap_{(N,v)}(\hat{\K}\cup Z) - \utopianGap_{(N,v)}(\hat{\K}),
	\end{equation}
	where $\hat{\K} \subseteq 2^N \setminus \K_0$.

	We have shown in Appendix~\ref{app:mere-technicality} that the utopian gap is
	\begin{equation}
		\label{eq: gap as scaled divergence}
		\utopianGap_{(N,v)}(\K) =
		\sum_{S \subseteq N}
		\gamma_S\left(\overline{v}_{\hat{\K}}(S) - \underline{v}_{\hat{\K}}(S)\right)-v(N),
	\end{equation}
	where $\gamma_S = \frac{|S|!(|N|-|S|)!}{|N|!}$.
	This equation can be viewed as a weighted $l_1$-norm of the difference between the upper and lower games.\footnote{Assuming they are viewed as vectors in $\mathbb{R}^{2^n}$.}

	Using this equality, we can rewrite Eq.~\eqref{app: eq: supermodularity constraints} as
	\begin{multline}
		\label{app: eq: supermodularity constraints2}
		\sum_{T \subseteq N}\gamma_T\left(\overline{v}_{\hat{\K}\cup S\cup Z}(T) - \underline{v}_{\hat{\K}\cup S\cup Z}(T) - \overline{v}_{\hat{\K}\cup S}(T) + \underline{v}_{\hat{\K}\cup S}(T)\right) \ge \\
		\sum_{T \subseteq N}\gamma_T\left(\overline{v}_{\hat{\K}\cup Z}(T) - \underline{v}_{\hat{\K}\cup Z}(T) - \overline{v}_{\hat{\K}}(T) + \underline{v}_{\hat{\K}}(T)\right).
	\end{multline}

	We obtain the result by setting $\K = \K_0 \cup \{\{j,k\}\}$, $S=\{i,j\}$, and $Z=\{k,l\}$.
	Denote $v(ij) = v(jk) - \varepsilon$ and $v(kl) = v(jk) + \delta$ for some $\varepsilon, \delta \ge 0$.
	For these values, we investigate each of the weighted terms separately.
	Let
	\begin{align}
		L_T & = \overline{v}_{\hat{\K}\cup S\cup Z}(T) - \underline{v}_{\hat{\K}\cup S\cup Z}(T) - \overline{v}_{\hat{\K}\cup S}(T) + \underline{v}_{\hat{\K}\cup S}(T), \\
		R_T & = \overline{v}_{\hat{\K}\cup Z}(T) - \underline{v}_{\hat{\K}\cup Z}(T) - \overline{v}_{\hat{\K}}(T) + \underline{v}_{\hat{\K}}(T).
	\end{align}
	In the reminder of the proof, we distinguish different cases based on the relation of $T$ and $\{i,j,k,l\}$ for all $T\in 2^N\setminus \K_0$.
	\begin{enumerate}
		\item $\{i,j,k,l\} \subseteq T\subsetneq N$:

		      Using Lemma~\ref{lem:change-of-bounds}, we get
		      \begin{align*}
			      L_T & = \underline{v}_{\K \cup S}(T) - \underline{v}_{\K\cup S\cup Z}(T), \\
			      R_T & = \underline{v}_{\K}(T) - \underline{v}_{\K\cup Z}(T),
		      \end{align*}
		      which can be expressed as
		      \begin{align*}
			      L_T & = v(jk) - v(jk) + \varepsilon - v(jk) - \delta, \\
			      R_T & = v(jk) - v(jk) - \delta.
		      \end{align*}
		      It holds $L_T = R_T + \varepsilon - v(jk) = R_T - v(ij)$.
		\item $\{i,j,k,l\}\subseteq N \setminus T$:

		      In this case, the lower bound on $v(T)$ is unaffected by the knowledge of $S,Z$, so
		      \begin{align*}
			      L_T & =  \overline{v}_{\K \cup S\cup Z}(T)
			      - \overline{v}_{\K \cup S}(T),                              \\
			      R_T & =  \overline{v}_{\K\cup Z}(T) - \overline{v}_{\K}(T).
		      \end{align*}
		      Since the only superset of $T$ in $\K\cup S \cup Z$ is $N$, the upper bound~\eqref{eq: upper game} reduces to
		      \begin{equation*}
			      \overline{v}_\mathcal{L}(T) =
			      v(N) - \underline{v}_\mathcal{L}(N\setminus T)=
			      1 - \underline{v}_\mathcal{L}(N\setminus T),
		      \end{equation*}
		      for every $\mathcal{L}\in \{\K,\K\cup S, \K \cup Z, \K \cup S\cup Z\}$. Thus
		      \begin{align*}
			      L_T & =  \underline{v}_{\K \cup S}(N\setminus T)- \underline{v}_{\K \cup S \cup Z}(N\setminus T), \\
			      R_T & =  \underline{v}_{\K}(N\setminus T)- \underline{v}_{\K \cup Z}(N\setminus T).
		      \end{align*}
		      Further, $i,j,k,l\in N\setminus T$, which means
		      \begin{align*}
			      L_T & =  v(jk)- v(jk)+\varepsilon - v(jk) - \delta, \\
			      R_T & =  v(jk) - v(jk)-\delta.
		      \end{align*}
		      As in the previous case, $L_T = R_T + \varepsilon - v(jk) = R_T - v(ij)$.

		\item $S=T$:

		      In this case, since $\underline{v}_{\mathcal{L}}(T) = \overline{v}_\mathcal{L}(T) = v(T)$ if $T \in \mathcal{L}$ for every $\mathcal{L}\in \{\K,\K\cup S, \K \cup Z, \K \cup S\cup Z\}$ we get
		      \begin{align*}
			      L_T & =  0,                                                                                                                               \\
			      R_T & =  \overline{v}_{\K\cup Z}(T)- \overline{v}_{\K}(T) = \underline{v}_{\K}(N \setminus T) - \underline{v}_{\K \cup Z}(N \setminus T).
		      \end{align*}
		      This means $R_T = 0 - v(jk) - \delta$, therefore it holds $L_T = R_T + v(jk) + \delta = R_T + v(kl)$.

		\item $Z=T$:

		      In this case, by the fact that $\underline{v}_{\mathcal{L}}(T) = \overline{v}_\mathcal{L}(T) = v(T)$ if $T \in \mathcal{L}$, it holds
		      \begin{align*}
			      L_T & =  \underline{v}_{\K\cup S}(T)- \overline{v}_{\K\cup S}(T), \\
			      R_T & = \underline{v}_{\K}(T)- \overline{v}_{\K}(T).
		      \end{align*}
		      Since $S \not\subseteq T$, it holds $\underline{v}_{\K \cup S}(T) = \underline{v}_\K(T)$, so we can rewrite
		      \begin{align*}
			      L_T & =  \underline{v}_{\K\cup S}(N \setminus T) - v(N) = v(jk)-\varepsilon - v(N), \\
			      R_T & = \underline{v}_{\K}(N \setminus T)- v(N) = -v(N).
		      \end{align*}
		      Thus it holds $L_T = R_T +v(jk) - \varepsilon = R_T + v(ij)$.

		\item $\emptyset \neq Z\cap T, Z \not\subseteq T$, and $T \not\subseteq Z$ or $\emptyset\neq S\cap T, S \not\subseteq T$, and $T \not\subseteq S$:

		      By Lemma~\ref{lem:bound-funcs-monotone}, and Lemma~\ref{lem:unaffected-bounds}, one immediately arrives at $L_T=R_T=0$.
	\end{enumerate}

	A necessary condition for supermodularity of the gap is \[\sum_{T\subseteq N}\gamma_{|T|} (L_T - R_T) \ge 0.\]

	This means, that if any upper bound $L$,
	\[\sum_{T\subseteq N}\gamma_{|T|} (L_T - R_T)\le L\]
	satisfies $L < 0$, the gap is not supermodular. To construct the upper bound, we focus on terms of the sum corresponding to those from the first and the second case above, i.e. $T \subseteq N$ s.t. $\{i,j,k,l\} \subseteq T \subsetneq N$ or $T \subseteq N \setminus \{i,j,k,l\}$. In these cases, $L_T - R_T \leq 0$. Since $\gamma_{|T|} \leq \gamma_{\floor{n/2}}$ for every $T \subseteq N$, we bound the corresponding terms by
	\[
		\gamma_{|T|}\left(L_T - R_T\right) \leq \gamma_{\floor{n/2}}\left(L_T - R_T\right).
	\]
	The upper bound thus sums to
	\[L = -\gamma_{\floor{n/2}}(c_1(n) + c_2(n)) v(ij) + \gamma_2(v(kl) + v(ij)),\]
	where $c_1(n) = 2^{n-4}-1$ and $c_2(n)=2^{n-4}-n+3$ is the number of terms satisfying conditions 1 and 2, respectively. Thus, $L < 0$ is equal to
	\[
		v(kl) < \left(\frac{\gamma_{\floor{n/2}}}{\gamma_2}\left(c_1(n) + c_2(n)\right)-1\right)v(ij),
	\]
	or since $c_1(n) + c_2(n) = 2^{n-3}-n+2$, it is equal to
	\begin{equation}
		v(kl) < \left[{\binom{n}{2}}{\binom{n}{\lfloor n/2 \rfloor}}^{-1}\left(2^{n-3}-n+2\right)-1\right]v(ij).
	\end{equation}
	The coefficient for $v(ij)$ is equal $2,\frac{28}{5},\frac{47}{5}$,$\dots$ for $n = 6,7,8,\dots$ and the sequence can be shown to be increasing. For $n=5$, the coefficient is equal to $0$, which makes the criterion useless as it translates to $0\le v(kl) < 0$.
\end{proof}
A corollary of the criterion is that, in any game where $|N| \ge 6$ and there exist $i,j,k,l \in N$ such that $v(ij) \leq v(jk) \leq v(kl)$ and $v(kl) < 2 v(ij)$
does not satisfy supermodularity of the gap. The criterion is satisfied by any game satisfying $v(ij)=v(jk)=v(kl)\neq 0$ for some $i,j,k,l$. Examples of totally monotonic, symmetric supperadditive, and graph games are immediate. Note that almost every symmetric superadditive ($v(ij)\neq 0$) as well as almost every additive game (different from zero game) has a non-supermodular gap.

\subsection{Simpler criterion for $|N| \geq 7$}
Using the same technique as in~\ref{app:gap-criterion}, one might derive weaker, however, simpler criterion. E.g. since $-\gamma_{\floor{n/2}}(c_1(n) + c_2(n)) \leq -\frac{1}{10}$ for every $n \geq 5$, it holds $L \leq -\frac{1}{10} v(ij) + \gamma_2\left(v(kl) + v(ij)\right)$, which is strictly smaller than $0$ if
$v(kl) < \left(1 -10\gamma_2\right)v(ij) / 10\gamma_2$,
or
\begin{equation}
	v(kl) < \frac{n^2 - n - 20}{20}v(ij).
\end{equation}
However, this criterion is useful only when $n\geq 7$, as for $n=5,6$, it holds $\alpha_n = \frac{n^2 - n - 20}{20}< 1$, which implies
\[0 \le v(ij) \leq v(kl) < \alpha_n v(ij),\]
thus $v(ij) = v(kl)=0$ and the criterion simplifies to $0 < \alpha_n 0$, which cannot be satisfied. However, for $n \geq 7$, it holds
\[\frac{5}{4} v(ij) \leq \alpha_n v(ij),\]
thus any game satisfying
\[0 < v(ij) \leq v(kl) < \frac{5}{4}v(ij)\]
cannot have a supermodular gap.

\end{document}